\documentclass[11pt]{article} 
\usepackage{latexsym, amsmath,
  amssymb,amsthm,amsopn,amsfonts,amscd, amsxtra} 
  \usepackage{indentfirst}
\usepackage{hyperref}
\usepackage{fourier,mathtools,commath}
\usepackage{epsfig,xcolor,graphicx,graphpap}
\usepackage[margin=1in]{geometry}
\usepackage{enumerate}
\usepackage[norelsize,linesnumbered,vlined,ruled,algo2e]{algorithm2e}
\usepackage{epstopdf}
\usepackage{mathrsfs}
\usepackage{framed}
\usepackage{verbatim}
\usepackage{algorithm}
\usepackage[noend]{algorithmic}
\usepackage{braket}
\usepackage[section]{placeins}
\usepackage{float}
\usepackage[version=4]{mhchem}
\usepackage{url}
\usepackage{xurl}

\newcommand{\calE}{{\mathcal E}} 
\newcommand{\CC}{{\mathbb C}}

\newcommand{\RR}{{\mathbb R}}

\theoremstyle{plain}

\newtheorem{theorem}{Theorem}
\newtheorem{prop}{Proposition}[section]

\newtheorem{lem}[prop]{Lemma}

\theoremstyle{definition}

\newtheorem{rem}{Remark}[section]

\def\squarebox#1{\hbox to #1{\hfill\vbox to #1{\vfill}}}

\newcommand{\wt}[1]{\widetilde{#1}}

\newcommand{\wh}[1]{\widehat{#1}}
\newcommand{\Or}{\mathcal{O}}

\DeclareMathOperator{\spanop}{span}
\IfFileExists{mathabx.sty}%
  {\DeclareFontFamily{U}{mathx}{\hyphenchar\font45}%
   \DeclareFontShape{U}{mathx}{m}{n}{<->mathx10}{}%
   \DeclareSymbolFont{mathx}{U}{mathx}{m}{n}%
   \DeclareFontSubstitution{U}{mathx}{m}{n}%
   \DeclareMathAccent{\widebar}{0}{mathx}{"73}%
}{%
  \PackageWarning{mathabx}{%
    Package mathabx not available, therefore\MessageBreak substituting
    widebar with overline\MessageBreak }%
  \newcommand{\widebar}[1]{\overline{#1}}%
}
\newcommand{\wb}[1]{\widebar{#1}}

\newcommand{\mc}[1]{\mathcal{#1}}

\newcommand{\ud}{\,\mathrm{d}}

\newcommand{\Hm}{\ce{H_2}}

\title{Symmetry Breaking and the Generation of Spin Ordered Magnetic States in Density Functional Theory due to Dirac
  Exchange for a Hydrogen Molecule}

\author{  
  Michael Holst \,\thanks{Department of Mathematics,
    University of California, San Diego,
    9500 Gilman Dr. 
    La Jolla, CA 92093, USA
    ({\tt mholst@math.ucsd.edu}).}
  \and
  Houdong Hu \,\thanks{Microsoft Corporation,
    Bellevue, WA, USA
    ({\tt vincehouhou@gmail.com}).}
  \and
  Jianfeng Lu\,\thanks{Department of Mathematics, Department of Physics, and
    Department of Chemistry,
    Duke University, 
    Box 90320, 
    Durham, NC 27708, USA
    ({\tt jianfeng@math.duke.edu}).}  
  \and
  Jeremy L.~Marzuola\,\thanks{Department of Mathematics,
    University of North Carolina at Chapel Hill, 
    CB\#3250 Phillips Hall, 
    Chapel Hill, NC 27599, USA
    ({\tt marzuola@math.unc.edu}).} 
  \and
  Duo Song\,\thanks{Physical and Computational Sciences Directorate, 
    Pacific Northwest National Laboratory, 
    Richland, WA 99354, USA
    ({\tt duo.song@pnnl.gov}).}
  \and
  John Weare\,\thanks{Department of Chemistry,
    University of California, San Diego,
    9500 Gilman Dr., 
    La Jolla, CA 92093, USA
    ({\tt jweare@ucsd.edu}).}}

\begin{document}    

\maketitle

\begin{abstract} 
We study symmetry breaking in the mean field solutions to the  {electronic structure problem} for the $2$ electron hydrogen molecule
 within the Kohn Sham (KS) local spin density {functional} theory with Dirac exchange (the XLDA model). This simplified model shows behavior related to that of the (KS) spin density functional theory (SDFT) predictions in condensed matter and molecular systems. The KS solutions to the constrained SDFT variation problem undergo spontaneous symmetry breaking {leading to the formation of spin ordered  states} as the relative strength of the non-convex exchange term increases.  {Numerically, we observe that with increases in the internuclear bond length the molecular ground state changes} from a paramagnetic state ({spin delocalized}) to an antiferromagnetic ({spin localized}) ground state and a symmetric delocalized  ({spin delocalized}) excited state. We further characterize the limiting behavior of the minimizer when the strength of the exchange term goes to infinity both analytically and numerically. This leads to further bifurcations and highly localized states with varying character. Finite element numerical results provide support for the formal conjectures. Various solution classes are found to be numerically stable. However,  {for changes in the $R$ parameter, numerical Hessian calculations demonstrate that these are stationary but not stable solutions.}
\newline

\end{abstract}

\section{Introduction}

In this paper, we report studies of the properties of density functional theory (DFT) energy minimizers within the context of the hydrogen molecule, {\Hm}. The (DFT) minimizers discussed  are related to those of  the Kohn-Sham spin density functional method. The exchange  correlation function \cite{Parr-Yang} is simplified by including only Dirac spin density exchange without correlation \cite{Parr-Yang} .  We will show that for fixed electron mass,
the structure of the minimizing Kohn-Sham solutions change character with the variation of parameters related to the relative strength of the exchange-correlation component of the
functional. {In particular, the changes of these parameters lead to bifurcations from globally stable delocalized product states with no spin localization to product states with electron spin localized on atomic sites (antiferromagnetic states). This behavior is similar to the formation of spin ordered states in the DFT analysis of highly correlated materials \cite{kubicki,Cox1992,Peng-Perdew2017,Roll2004}.} 

Similar studies varying the molecular bond length {have been} undertaken using robust finite
element methods for Hartree-Fock and SLDA functionals in
\cite{HoudongThesis} and for Hartree-Fock 
using a maximum overlap method in \cite{barca2014communication}. {The precise electron configurations that occur in the ground states of 
 such problems is important for}
further developing density functional theory (see e.g., \cite{Yang:08,Yang:12}) as well as directly in the application of density functional theory (DFT) to highly correlated condensed materials  \cite{kubicki,Cox1992,Peng-Perdew2017,Roll2004} and in spin ordered molecular systems. In addition, uniqueness and symmetry breaking in other
quantum mechanical models  have recently been studied widely in for
instance the works \cite{FL2,FL1} for polaron models,
\cite{gontier2018lower,gontier2018spin,GriesemerHantsch:12} for Hartree-Fock models of atoms, and many others.  A similar strategy to that undertaken here in one of our limits was explored for the periodic Thomas--Fermi--Dirac--von Weizs\"acker model in \cite{ricaud2017,ricaud2017symmetry}.

We consider a neutral hydrogen molecule {\Hm} with nuclei placed $2R$
apart. The external potential is given by (after a possible coordinate
change){
\begin{equation}\label{eq:defVR}
  V_R(x) = - \frac{1}{\abs{x - R e_1 }} - \frac{1}{\abs{x + R e_1}},
\end{equation}
where $e_1$ is the $(1, 0, 0)$ vector in $\RR^3$.} The two-electron Schr\"odinger operator is given by 
\begin{equation}
  H_2 = - \frac{1}{2} \Delta_{x_1} - \frac{1}{2} \Delta_{x_2} + V_R(x_1) + V_R(x_2) + \frac{1}{\abs{x_1 - x_2}},
\end{equation}
{where $x_1$ and $x_2$ denote the position of the two electrons in the system. Atomic units are used throughout.}

In this work, we will consider the spin-polarized density functional
theory with the exchange energy taken to be Dirac exchange and
without correlation energy. In the literature, the spin free version of this model is sometimes
referred as the XLDA model. We are interested in the spin paired ground state of the system with spin up $\psi_+$ and spin down $\psi_-$  spatial wave functions. 
We would like to study the impact of the
exchange term on the electronic structure. Therefore, {we introduce
a strength parameter $\alpha \geq 0$ for the exchange energy functional}. The
DFT energy functional is hence given by
\begin{multline}
  \label{eq:LDA}
  \calE_{\alpha}(\psi_+, \psi_-) = \frac{1}{2} \int \abs{\nabla \psi_+}^2
  \ud x + \frac{1}{2} \int \abs{\nabla \psi_-}^2 \ud x
  +   \int V_R(x) \rho(x) \ud x \\
  + \frac{1}{2} \iint \frac{\rho(x) \rho(y)}{\abs{x-y}} \ud x \ud y -
  \alpha \int \left( \abs{\psi_+}^{8/3} + \abs{\psi_-}^{8/3} \right)
  \ud x,
\end{multline}
where the electron density of the system is given by
\begin{equation}
  \rho(x) = \abs{\psi_+(x)}^2 + \abs{\psi_-(x)}^2.
\end{equation}
 Note in particular in \eqref{eq:LDA} the Dirac exchange term is spin-polarized: let $\rho_{\pm} = \abs{\psi_{\pm}}^2$ {be the spin-polarized densities}, the exchange term is given by 
\begin{equation}
  -\alpha \int \bigl( \rho_+^{4/3} + \rho_-^{4/3} \bigr) \ud x
\end{equation}
as the exchange effect originating from Pauli's exclusion principle only 
occurs between electrons with same spin polarization \cite{Oliver-Perdew,Parr-Yang}.

{The potential $V_R$ defined in \eqref{eq:defVR} corresponds to the {\Hm} molecule having reflection symmetry.}  We are interested in
the symmetry (delocalization) (or lack of symmetry, localization) of
$\psi_+ $ and $ \psi_-$. {In particular, when a solution
  inherits the symmetry of the potential $V_R$, the electron wave
  functions will be even split across both atoms, hence we refer to
  that state as a delocalized state.  Otherwise, each electron wave
  function will be supported on one particular atom in the molecule,
  in which case we call the electrons localized.  We note here that we
  are considering the only the breaking of spatial symmetry of the
  wave functions among neutral spin minimizers (i.e., spin singlet
  configurations).}  We call the minimizer with the symmetry
constraint, $\psi_+ = \psi_- = \psi_R$, a restricted minimizer to the
energy functional, denoted as $\psi_R$.  Thus
\begin{equation}\label{eq:RDFT}
  \begin{aligned}
    & \psi_R = \arg\min\; \calE_{\alpha}(\psi, \psi) \\
    & \qquad \text{s.t. } \quad \int \abs{\psi}^2 = 1. 
  \end{aligned}
\end{equation}
The unrestricted minimization on the other hand considers all possible
$\psi_+$ and $\psi_-$ with the normalization constraints. To
distinguish, we denote the minimizers as $\psi_{\pm}$.
\begin{equation}\label{eq:UDFT}
  \begin{aligned}
    & (\psi_+, \psi_-) = \arg\min \calE_{\alpha}(\psi_+, \psi_-) \\
    & \qquad \text{s.t. } \quad \int \abs{\psi_+}^2 = \int
    \abs{\psi_-}^2 = 1.
  \end{aligned}
\end{equation}
Our goal in this work is to understand the symmetry breaking, i.e.,
the question whether $\psi_{+} = \psi_- = \psi_R$. The following
result gives the existence of minimizers to both \eqref{eq:RDFT} and \eqref{eq:UDFT}.

\begin{prop}
  For all $\alpha \geq 0$, there exist solutions $(\phi_+, \phi_-) \in
  H^1 \times H^1$ with $\int |\phi_\pm|^2 \ud x = 1$ such that
  \begin{equation*}
    \calE_\alpha (\phi_+,\phi_-) = \min_{\psi_{\pm} \in H^1; \int |\psi_\pm|^2 dx = 1} \calE_\alpha(\psi_+, \psi_-).
  \end{equation*}
\end{prop}

For a proof of this proposition, we refer the reader to the Concentration Compactness tools employed in \cite[Theorem $1$]{anantharaman2009existence} or specifically for LDA models the recent work of \cite{gontier2015existence}, where a general existence theory is addressed for LDA models of this type with neutral or positive charge.  

For the energy functional \eqref{eq:LDA}, we have two parameters $R$
and $\alpha$ in the functional. We expect the following behavior of
the minimizers for different ranges of parameters:
\begin{enumerate}
\item For $\alpha = 0$ and any $R > 0$, the minimizer have the
  symmetry $\psi_+ = \psi_- = \psi_R$.
\item Fix $R \geq 0$, when we increase $\alpha$ from $0$: The
  minimizer is initially symmetric (hence it is continuous at
  $\alpha = 0$), the symmetry is broken for larger $\alpha$
  ($\psi_+ \neq \psi_-$). The critical parameter $\alpha$ for the
  transition from symmetric to asymmetric minimizer depends on $R$.
\item Fix $\alpha > 0$, for $R$ sufficiently large, the minimizer
  is asymmetric. 
\end{enumerate}
Therefore, this suggests a two-dimensional phase diagram where the
axes are $R$ and $\alpha$ with a phase transition from symmetric to
asymmetric minimizers.  In the current manuscript, we will fix $R$ and
vary the parameter $\alpha$ in our analysis {in order to
  prove that symmetry breaking occurs in the $\alpha$ parameter as
  predicted}.  However, we will demonstrate the $(\alpha, R)$ phase
diagram numerically {and hence lend numerical support to
  the conjectured behaviors in $R$}. Some technical difficulties arise
in the the analysis when varying $R$ {in particular when
  taking $R \to \infty$}, which we comment on in
Section~\ref{sec:conclusion} and plan to address in future work.  We
make our statements precise in the following theorems.
\begin{theorem}
\label{thm1}
  Fix $R > 0$, denote $\psi_R$ the minimizer of \eqref{eq:RDFT} and
  $\psi_{\pm}$ the minimizer of \eqref{eq:UDFT}, we have $\psi_{\pm} =
  \psi_R$ for $\alpha \ll 1$, and $\psi_+ \not= \psi_-$ for $\alpha
  \gg 1$.
\end{theorem}

In other words, as we increase $\alpha$, the symmetry
$\psi_+ = \psi_-$ is broken. In fact, we can give a more precise
characterization of the minimizer $\psi_{\pm}$ as
$\alpha \to \infty$.

\begin{theorem}
\label{thm2}
  Fix $R > 0$, as $\alpha \to \infty$, {up to symmetries of the equation}, the rescaled and translated minimizer of \eqref{eq:UDFT} 
\begin{equation}
\label{eqn:minform}
  \alpha^{-\frac32} \psi_\pm ( \alpha^{-1} (x \mp  R e_1)) 
\end{equation}
converges to $\phi$ in $H^1$, where $\phi$ is the unique positive,
radial solution to the equation
\begin{equation}
-\frac12 \Delta \phi - \frac43 |\phi|^{\frac23} \phi + E \phi = 0, \ \text{with} \  \int |\phi|^2 dx = 1.
\end{equation}
This can also be seen as the constrained minimizer of the Lagrangian
\begin{equation}
  \label{eq:semilinearLag}
  \calE_{s}(\phi) = \frac{1}{2} \int \abs{\nabla \phi}^2
  \ud x - \int \abs{\phi}^{8/3} 
  \ud x,
\end{equation}
with mass $\int \abs{\phi}^2 \ud x = 1$.  In other words, as $\alpha
\to \infty$, each electron becomes concentrated over a different
nucleus.
\end{theorem}

Our results in the small $\alpha$ setting rely heavily on the results
of Lieb, Lions and others relating to the concentration compactness
phenomenon for constructing minimizers of constrained Lagrangians at
$\alpha =0$, then an application of the Implicit Function Theorem for
small $\alpha$.  Our result for large $\alpha$ on the other hand
follows from essentially comparing the variational problem to a
scale-invariant semi-linear problem, which in turn relies strongly on
the orbital stability of solitons for the unperturbed Dirac
nonlinearity in $3$ dimensions, $|u|^{\frac23} u$.  

The proof for the small $\alpha$ regime is presented in
Section~\ref{sec:small}, while the large $\alpha$ regime is treated in
Section~\ref{sec:large}.  We present the analysis in detail for fixed
$R > 0$ and varying $\alpha$ throughout the proof.  Without loss of
generality, for our analysis we will assume $R = 1$ and denote
$V = V_R$.  Detailed numerical studies of the $(\alpha,R)$ phase
diagram and in particular the transition between small and large
$\alpha$ for fixed $R$ are discussed in Section \ref{sec:numresults}.
Concluding remarks and a discussion of the analysis in the case of varying
$R$ are included in Section \ref{sec:conclusion}.  The numerical
methods are presented in Appendix~\ref{sec:numerics} using a finite
element package developed by the group of the first author and
implemented in the thesis of the second author to study variational
problems in electronic structure theory.

\section{Proof in the small $\alpha$ regime}
\label{sec:small}

\subsection{The restricted Hartree model: Case $\alpha = 0$}

When $\alpha = 0$, the energy functional we consider becomes
\begin{equation}
\label{eq:LDA0}
\calE_{0}(\psi_+, \psi_-) = \frac{1}{2} \int \abs{\nabla
  \psi_+}^2 \ud x + \frac{1}{2} \int \abs{\nabla \psi_-}^2 \ud x
+ \int V(x) \rho(x) \ud x 
+ \frac{1}{2} \iint \frac{\rho(x) \rho(y)}{\abs{x-y}} \ud x \ud y. 
\end{equation}

Without the exchange-correlation energy, the minimizer is always
symmetric. Indeed, fixing any density $\rho$ with $\int \rho = 2$, we
have
\begin{equation}
  \calE_{0}(\sqrt{\rho}/\sqrt{2}, \sqrt{\rho}/\sqrt{2}) 
  = \inf \bigl\{\calE_{0}(\psi_+, \psi_-) \mid 
  \abs{\psi_+}^2 + \abs{\psi_-}^2 = \rho \bigr\}.
\end{equation}
Define $\rho_+ = \abs{\psi_+}^2$ and $\rho_- = \abs{\psi_-}^2$, the
above follows from the convexity
\begin{equation}
  2 \int \abs{\nabla \sqrt{(\rho_+ + \rho_-) / 2}}^2 \ud x 
  \leq \int \abs{\nabla \sqrt{\rho_+}}^2 \ud x 
  + \int \abs{\nabla \sqrt{\rho_-}}^2 \ud x, 
\end{equation}
and the equality holds if and only if $\rho_+ = \rho_-$ (see
\cite[Page 177, Theorem 7.8]{LiebLoss:01}). Thus, we may denote the
common orbital function as $\phi = \psi_+ = \psi_-$, which minimizes the functional 
\begin{equation}\label{eq:calE}
  \calE_0(\phi) = \int \abs{\nabla \phi}^2 \ud x + 2 \int V(x) \abs{\phi}^2 \ud x
  + 2 \iint \frac{\abs{\phi(x)}^2 \abs{\phi(y)}^2}{\abs{x - y}} \ud x \ud y. 
\end{equation}
Note that this functional has the same form as the restricted Hartree
model treated in \cite[Theorem II.2]{lions1987solutions}, which
guarantees the existence of a minimizer. Moreover, the minimizer is
non-negative without loss of generality and satisfies the
Euler-Lagrange equation
\begin{equation}
\label{alpha0}
  - \frac12 \Delta \phi + E_0 \phi +  V \phi  + 2 \bigl( v_c  
\ast \abs{\phi}^2 \bigr)\, \phi = 0
\end{equation}
where $v_c(x) = \abs{x}^{-1}$ denotes the Coulomb kernel and $E_0 \geq
0$ is the Lagrange multiplier.  We now show that $E_0$ must be
strictly positive. Suppose $E_0 = 0$, define $W := V + 2 v_c \ast
\abs{\phi}^2$, we have
\begin{equation}
  - \frac{1}{2} \Delta \phi + W \phi = 0. 
\end{equation}
Using Newton's theorem, the spherical average of $W$, denoted by
$\wb{W}$ is non-positive outside the ball $B_R$ (since the ball
contains all the nuclei charge). Thus, we get trivially that the
positive part of $\wb{W}_+ = \max\{\wb{W}, 0\} \in
L^{3/2}(B_R^c)$. This implies that $\phi \not\in L^2(B_R^c)$ by
 \cite[Lemma 7.18]{lieb1981thomas}, which is clearly a contradiction,
since $\int \abs{\phi}^2 = 1$. Therefore, $E_0 > 0$.  This implies that the nuclear potential 
is properly binding in a similar sense to that explored in \cite{RuskaiStillinger:84}.

For a purpose that will be clear later, we also consider the variational
problem \eqref{eq:calE} with more general mass constraints and denote
the minimum as $I_{M}$: 
\begin{equation}
  \begin{aligned}
    I_{M} & := \inf \Bigl\{\mc{E}_0(\phi) \mid \int
    \abs{\phi}^2 = M / 2 \Bigr\} \\
    & = \inf \Bigl\{\mc{E}_0(\phi) \mid \int
    \abs{\phi}^2 \leq M / 2\Bigr\} \\
    & = \inf \Bigl\{\mc{E}_0(\sqrt{\rho}/\sqrt{2}) \mid \int \rho
    \leq M \Bigr\},
  \end{aligned}
\end{equation}
where the second equality follows from the fact that $I_{M}$ is
monotonically decreasing {in $M$} as we can always put some excessive charge
far away from the nuclei with negligible contribution to the
energy. Furthermore, $I_{M}$ is strictly convex for $M \in
[0, M_c)$ for some $M_c \geq 2$, which follows the
standard convexity argument applies to
$\mc{E}_0(\sqrt{\rho}/\sqrt{2})$ as in the proof of parts (iii) and
(iv) of \cite[Corollary II.1]{lions1987solutions} (see also the proof
of convexity of the energy of the related Thomas-Fermi-von
Weizs\"acker theory in \cite{benguria1981thomas}). We also have the
relation
\begin{equation}
  \frac{\partial I_{M}}{\partial M} \Big\vert_{M = 2} = - E_0 < 0,  
\end{equation}
since $E_0$ is the Lagrange multiplier corresponding to the constraint
{$\int \abs{\phi}^2 = M/2$}. This in turn guarantees that $M_c
> Z$, as in Part (i) of \cite[Corollary II.1]{lions1987solutions}. Therefore, denote $E_0(M)$ the corresponding Lagrange multiplier for $I_{M}$, we arrive at 
\begin{equation}\label{eq:derE0}
  - \frac{\partial E_0(M)}{\partial M} \Big\vert_{M = 2} = \frac{\partial^2 I_{M}}{\partial M^2} \Big\vert_{M = 2}  > 0. 
\end{equation}
Following the analysis of \cite[Theorem $3.1$]{lieb1977hartree} using
elliptic estimates, one observes that if $\phi \in H^1$ is a solution to
\eqref{alpha0}, then $\phi \in H^2$.

\subsection{Implicit function theorem analysis for small $\alpha$}

We consider \eqref{eq:LDA} for $\alpha$ small.  First of all, by
restricting to the class of solutions symmetric with respect to
reflection in $x$, we know there exists a delocalized solution obeying
the correct symmetry properties.  For $\alpha = 0$, \eqref{eq:LDA} is
a convex functional and there exists a \emph{unique} delocalized
solution $ \psi_+ = \psi_- = \phi$ such that $\| \phi \|_{L^2} =
1$. The following result extends the uniqueness to small $\alpha$.  This result is similar to one proved in \cite{le1993some} for the Thomas-Fermi-Dirac-von Weizs\"acker model.

\begin{prop}
\label{prop:smallalpha}
For {$\alpha \geq 0$} sufficiently small, there exists a unique, delocalized minimizer to \eqref{eq:LDA} {such that $\psi_+ = \psi_- = \phi$} with $\| \phi \|_{L^2} = 1$.  The dependence upon $\alpha$ is $C^1$.  
\end{prop}
  
The remainder of this section is devoted to the proof of Proposition
\ref{prop:smallalpha}.  The idea is to construct a symmetric solution
branch stemming from the unique solution at $\alpha = 0$ that comes
from the convexity of the energy functional in that limit.  While the
positive $\alpha$ perturbation is non-convex, the Euler-Lagrange
equations can be solved using a Lyapunov-Schmidt reduction. In fact,
we will see that we can construct an implicit function theorem
argument using the convexity at $\alpha = 0$ and in doing so, that
locally only the symmetric branch will be possible.  First, we will
allow the branch to vary with respect to mass, then we will fix the
Lagrange multipliers $E_+$ and $E_-$ (in most cases we will observe $E_+ = E_-$) as a function of $\alpha$ to guarantee the
mass $1$ electron branch. 
 
The Euler-Lagrange equations for $\calE_{\alpha}$ {can be
  written as the following with $F$ defined as a function on
  $(H^2)^2 \times \RR^3$} 
  \begin{equation}\label{eq:defF}
  F(\psi_+, \psi_-; \alpha, E_+,E_-) :=
  \begin{pmatrix}
    - \frac12 \Delta \psi_+ + E_+ \psi_+ + V \psi_+   +  \bigl( v_c \ast ( \abs{\psi_+}^2 + \abs{\psi_-}^2 ) \bigr)\, \psi_+ - \frac43 \alpha \abs{\psi_+}^{\frac23} \psi_+ \\
    - \frac12 \Delta \psi_- + E_- \psi_- + V \psi_- + \bigl( v_c
    \ast ( \abs{\psi_+}^2 + \abs{\psi_-}^2 ) \bigr)\, \psi_- - \frac43 \alpha
    \abs{\psi_-}^{\frac23} \psi_-
  \end{pmatrix} = 0.
\end{equation}
To apply the Lyapunov-Schmidt reduction, we need to address the kernel
of the Jacobian with respect to $\psi_\pm$ for the Euler-Lagrange
equations. This is given by the operator
\begin{align}
\label{DF}
& D_{\psi} F (\psi_+, \psi_-; \alpha, E_+, E_-) 
\begin{pmatrix}
  f_+ \\
  f_- 
\end{pmatrix} 
=  
\begin{pmatrix} 
  L_{+} f_+ + 2 \psi_+ v_c \ast (\psi_- \wb{f_-})  \\
  L_{-} f_- + 2 \psi_- v_c \ast (\psi_+ \wb{f_+})
\end{pmatrix}
\end{align}
for
\begin{align*}
  L_{+}(\psi_\pm;  \alpha, E_+) f_+ & := \Bigl( - \frac12 \Delta  + E_+ + V    +   v_c \ast ( \abs{\psi_+}^2 + \abs{\psi_-}^2 ) - \frac{20}{9} \alpha \abs{\psi_+}^{\frac23} \Bigr) f_+ 
  + 2 \psi_+ v_c \ast (\psi_+ \wb{f_+}); \\
  L_{-} (\psi_\pm; \alpha, E_-) f_- & := \Bigl( - \frac12 \Delta  + E_- + V    +    v_c \ast ( \abs{\psi_+}^2 + \abs{\psi_-}^2 ) - \frac{20}{9}\alpha \abs{\psi_-}^{\frac23} \Bigr) f_-
  + 2 \psi_- v_c \ast (\psi_- \wb{f_-}).
\end{align*}
For $\psi_+ = \psi_- = \phi$, the unique solution at $\alpha = 0$ with
resulting Lagrange multiplier $E_0$ stemming from the convexity of
$\calE_{0}$, we have
\begin{align*}
  D_\psi  F (\phi, \phi; 0, E_0,E_0) \begin{pmatrix} f_+ \\ f_-
\end{pmatrix} = 
\begin{pmatrix}
  L_{\phi, E_0} f_+ + 2 \phi v_c \ast (\phi \wb{f_-})    \\
  L_{\phi, E_0} f_- + 2 \phi v_c \ast (\phi \wb{f_+}) 
\end{pmatrix}
\end{align*}
for
\begin{equation*}
  L_{\phi, E_0} f   = \Bigl( - \frac12 \Delta  + E_0 + V    + 2 v_c \ast \abs{\phi}^2 \Bigr) f + 2 \phi v_c \ast ( \phi \wb{f} ).
\end{equation*}
In the class of symmetric solutions, the problem reduces to solving a scalar equation instead of a system of equations. 

\subsection{Analysis of the linearized operators for $\alpha = 0$}

We prove here that at $\alpha = 0$, $\psi_+ = \psi_- = \phi$, $E = E_0$, then the linearized operator has a kernel, but it can only lead to solutions where $\psi_+$ and $\psi_-$ take on different masses.  This is a key step in applying the implicit function theorem in $\alpha$ locally.  To see this, we linearize \eqref{alpha0} to get the operator
\begin{equation}
\label{linalpha0}
\wt{L}_{\phi,E_0} f   = \Bigl( - \frac12 \Delta  + E_0 + V    + 2 v_c \ast \abs{\phi}^2 \Bigr) f + 4 \phi v_c \ast ( \phi \wb{f} ).
\end{equation}
We observe that the operator $\wt{L}_{\phi, E_0}$ can be written in the form
\begin{equation*}
  \wt{L}_{\phi, E_0}  =  - \frac12 \Delta  + E_0 + V    + V_\phi + W_\phi,
\end{equation*}
where
\begin{equation*}
  V_\phi := 2 v_c \ast \abs{\phi}^2 
\end{equation*}
is a potential with $1/|x|$ decay and 
\begin{equation*}
  W_\phi f := 4 \phi v_c \ast ( \phi \wb{f} )
\end{equation*}
is a self-adjoint {convolution operator, where we note
  that the function $\phi$ is exponentially decaying}.  
  

Since $V + V_\phi + W_\phi$ is a relatively compact perturbation, we
observe that the continuous spectrum of $\wt{L}_{\phi, E_0} $ is the
interval $[E_0,\infty)$ by applying Weyl's Theorem, see
\cite{RS4,LiebLoss:01} for instance, or \cite{Lenzmann} where the
functional analysis of Hartree-style equations is discussed in some
detail.

\begin{lem}
\label{l:kernel}
The operator $ \wt{L}_{\phi, E_0}  $ has only trivial kernel.
\end{lem}

\begin{proof}
Let us assume to the contrary there exists $f \in H^2$ such that
\[   \wt{L}_{\phi, E_0}  f = 0 .\]
Then, we observe that
\[ 0 =  \bigl\langle ( - \tfrac12 \Delta  + E_0 + V    + V_\phi) f, f \bigr\rangle + \bigl\langle W_\phi f, f \bigr\rangle. \]
{Given the structure of $W_\phi$ and the nature of the state $\phi$, we have that
\[ \langle W_\phi f ,f \rangle > 0 \ \ \text{for} \ \ f \neq 0, \]
since the Coulomb kernel is strictly positive, which is easily seen from the Fourier representation.}

{Taking the orthogonal decomposition $f = c \phi + \phi^\perp$ and
using that $\phi$ is the unique kernel of the operator
$ L_- = - \frac12 \Delta + E_0 + V + V_\phi$, where the notation $L_-$ here is chosen to match that of the semilinear literature for linearized operators about nonlinear states, see for instance \cite{Lenzmann}.  Hence if $\phi^\perp \neq 0$, we have a contradiction immediately from
the coercivity of the operator $L_-$.  Thus, it remains the
  possibility that $f = c \phi$. However if $c \neq 0$, we have
  $\bigl\langle W_\phi f, f \bigr\rangle > 0$, and therefore $f = 0$.}
\end{proof}

\begin{rem}
  This is a similar strategy to that of standard semi-linear problems,
  however in such a case the perturbation of the $L_-$ operator is
  negative in total and hence the spectral theory of the linearized
  operator must be understood in much greater detail.  Here, the
  perturbation is actually positive, so the arguments are greatly
  simplified.  {Also, we have a potential $V$ here, which
    has broken the translation invariance and hence we do not need to
    consider a $1$-parameter family of functions, but just a single
    $\phi$ in the kernel of $L_-$.}
\end{rem}

\subsection{Construction of solutions near $\alpha = 0$}

\begin{prop}
\label{prop:sym}
  {The Jacobian $ D_\psi F (\phi, \phi; 0, E_0, E_0) $ as
    defined in \eqref{DF} has kernel given by
    $\spanop \{(\phi, -\phi) \}$. As a result, there exists a unique
    $C^1$ path of solutions in $(\alpha, E_+, E_-)$ for the equation
    \eqref{eq:defF} with fixed constraint $\| \psi_\pm \|_{L^2} = 1$
    in $H^2 \times H^2$. Moreover, the unique solution satisfies the
    symmetry $\psi_+ = \psi_-$.  } 
\end{prop}

\begin{proof}
  We must study the invertibility of $D_\psi F$ at $\alpha = 0$.  In the restricted
  space, $\psi_+ = \psi_-$, the invertibility is established in Lemma
  \ref{l:kernel} through the invertibility of {
    $\wt{L}_{\phi, E_0}$}.  {More generally, let us
    consider $(f_1, f_2)$ that solves $D_{\psi} F (f_1, f_2)^{\mathrm{T}} = 0$,
    then $\wt{L}_{\phi, E_0} (f_1 + f_2) = 0$.}  Hence, either
  $f_1 + f_2$ is a non-trivial kernel function of $\wt{L}_{\phi, E_0}$
  {(which is excluded by Lemma~\ref{l:kernel})} or
  $f_1 = - f_2 = f$ and $f$ is a non-trivial kernel function for a
  modified operator
  \[ \mathcal{L}_{\phi,E_0} f = \Bigl( - \frac12 \Delta + E_0 + V + 2
    v_c \ast \abs{\phi}^2 \Bigr) f , \] which through the equation
  satisfies $\mathcal{L}_{\phi, E_0} \phi = 0$.  Since $\phi > 0$, it
  is the ground state and simple. {Therefore the kernel of
    $D_\psi F (\phi, \phi; 0, E_0, E_0) $ is one dimensional and described completely as
    $\spanop \{(\phi, -\phi) \}$.  }

    The remaining proof relies on varying $E_\pm$ using
    the standard Lyapunov-Schmidt construction of $\psi_\pm (E, \alpha)$
    solving the Euler-Lagrange equation,
    {see \cite[Proposition 1]{kirr2011symmetry} for a
      general discussion of the method. We write }
    \[ ( \psi_+, \psi_-) = (\phi, \phi) + c_0 ( \phi, - \phi) +(
      \eta_+, -\eta_-), \] for $c_0 \sim \sqrt{\alpha}$ and 
    $\int (\eta_+ + \eta_-) \phi \ud x = 0$. 
    We claim that
    \[
      \| \eta_{\pm} (c_0, E_{\pm}-E_0, \alpha) \|_{H^2} \lesssim
      c_0^2, \ \ | E_{\pm} - E_0 | \lesssim c_0^2, \] {
      where the dependence of $\eta_{\pm}$ upon our bifurcating parameters have been made
      explicit and in particular we have shifted the dependence upon $E_\pm$ to that $E_{\pm} - E_0$ such that $\eta (0,0,0) = 0$ for simplicity. } Indeed, expanding \eqref{eq:defF} about
    $(\phi, \phi)$ in this fashion we have    
    \begin{align}
      \label{eqn:eta1}
      F(\phi, \phi; 0, E_0, E_0 ) + D_{\psi}F( \phi, \phi; 0, E_\pm) 
      \begin{pmatrix}
        \eta_+ \\
        \eta_-
      \end{pmatrix} + \mathcal{R} (\eta, c_0, \alpha,E_\pm - E_0,\phi) 
 = \begin{pmatrix}
        0 \\
        0
      \end{pmatrix},
    \end{align}
    where the remainder term $  \mathcal{R}  (\eta, c_0, \alpha,E_\pm - E_0,\phi) $ will be specified below and satisfies
    \[
| \mathcal{R}  | \leq  C ( |\eta|^2 + c_0^2 + \alpha + (E_\pm - E_0)).
\]
    We note here that the linearization is $D_{\psi}F( \phi, \phi; 0, E_\pm)$
    and not $D_{\psi}F( \phi, \phi; 0, E_0, E_0)$.  Using the properties of $D_{\psi}F$,
    $\phi$ and $E_0$, we then observe that we can first re-write the equation for
   \[ 
   (\eta_+, \eta_-) = (\eta_+, \eta_-) (c_0, E_\pm-E_0,\alpha) 
   \]
   as
    \begin{align}
      \label{eqn:eta2}
      \begin{pmatrix} 
        \eta_+ \\
        \eta_-
      \end{pmatrix} = (P^\perp D_{\psi} F (\phi, \phi; 0, E_\pm) P^\perp)^{-1}
      P^\perp  \mathcal{R}  (\eta, c_0, \alpha,E_\pm - E_0) ,
    \end{align}
    where
    $P^\perp \vec f = \vec f - \bigl\langle \vec f, ( \phi, -
    \phi)^{\mathrm{T}} \bigr\rangle ( \phi, - \phi)^{\mathrm{T}} $ for
    $\vec f=(f_1, f_2)^T$ projects to the orthogonal complement of the
    kernel of $D_{\psi} F(\phi, \phi, 0, E_0, E_0)$.   {This implies
    that \eqref{eqn:eta1} can be written as
    \begin{align*}
      \left( - \frac12 \Delta + E_\pm  + V \right) \eta_{\pm} & = ( E_\pm - E_0) \phi + c_0 ( E_\pm - E_0) \phi - \alpha | \phi (1 \pm c_0) + \eta_\pm |^{\frac23 } (\phi (1 \pm c_0) + \eta_\pm) \\
                                                              & + v_c * \left[  2 c_0^2 \phi^2 + 2 (1+c_0) \phi \eta_+ + 2 (1-c_0) \phi \eta_-  + \eta_+^2 + \eta_-^2     \right] \phi \\
                                                              &  \pm  c_0 v_c * \left[  2 c_0^2 \phi^2 + 2 (1+c_0) \phi \eta_+ + 2 (1-c_0) \phi \eta_-  + \eta_+^2 + \eta_-^2     \right] \phi \\
                                                              & + v_c * \left[   2 \phi^2 + 2 c_0^2 \phi^2 + 2 (1+c_0) \phi \eta_+ + 2 (1-c_0) \phi \eta_-  + \eta_+^2 + \eta_-^2    \right] \eta_{\pm} .
    \end{align*}
    As we are taking $|E_\pm - E_0|$ small,
    $P^\perp D_{\psi}F (\phi, \phi; 0, E_\pm) P^\perp$ is invertible since
    $P^\perp D_{\psi}F (\phi, \phi; 0, E_0) P^\perp$ is invertible. We
    observe that by the Implicit Function Theorem there exists a solution $\vec \eta = (\eta_-,\eta_+)$ such that 
    \begin{equation}
    \label{etabd}
      \| \vec{\eta} \|_{H^2} \leq C \bigl(c_0^2 + \alpha+ |E_+ - E_0| + |E_- - E_0|  \bigr).
    \end{equation}
 } 

Projecting \eqref{eqn:eta1} onto $(\phi, -\phi)$, we compute directly that\footnote{
  Here, we use the following bound pointed out to the authors by
  N. Visciglia: Let $\alpha>0$ be given. Then for every $a, b\in \CC$
  we have the following inequality:
\begin{equation*}
  \bigl\lvert (a+b)|a+b|^\alpha-a|a|^\alpha-b |b|^\alpha  \bigr\rvert \lesssim 
  (|a| |b|^\alpha + |b| |a|^\alpha).
\end{equation*}  
}
\begin{equation}
\label{Epm:eqn}
  (E_+ - E_0) - (E_- - E_0) + c_0 [ (E_+ - E_0) + (E_- - E_0) ] +
  \mathcal{O} \Bigl( \| \vec{\eta} \|_{L^2}^2 + c_0^3 + \alpha \bigl(
  c_0 + \| \vec{\eta} \|_{L^2}^{\frac23} \bigr) \Bigr) = 0,
\end{equation} 
where we have implicitly used uniform Sobolev bounds on $\phi$ and the smallness of $c_0$, $\alpha$ and $E_\pm - E_0$.  
This allows us to use the Implicit Function Theorem once again to
solve for $E_+$ given $E_-$ and observe that
\[ | E_+ - E_0 | \leq C\bigl( c_0^3 + \alpha c_0  + \alpha^{\frac53} + |E_- - E_0|\bigr), \]
where we have plugged in the bound on $\| \eta \|_{L^2}$ from \eqref{etabd}.  
Now, using the two constraint equations for the mass, we have
\begin{equation}
\int (2 c_0 \phi + c_0^2 \phi^2 + 2( 1+ c_0) \phi \eta_+ + \eta_+^2) dx = 0 
\label{coneq1}
\end{equation}
and
\begin{equation}
\int (-2 c_0 \phi + c_0^2 \phi^2 + 2( 1- c_0) \phi \eta_- + \eta_-^2) dx = 0. 
\label{coneq2}
\end{equation}
{Using the linear combination \eqref{coneq1} $+$ \eqref{coneq2} and the orthogonality of $(\eta_+,\eta_-)$ to $(\phi,\phi)$ as constructed, we first observe by plugging the resulting implicit bound on $\| \vec \eta\|_{L^2} \leq C c_0^2$ in \eqref{Epm:eqn} that
\[ |E_- - E_0| \leq C (\alpha + c_0^2 ), \]
which implies for the linear combination \eqref{coneq1} $-$ \eqref{coneq2}, we can observe that 
\[ c_0^2 \leq C \alpha. \] 
}

Once the overall dependence upon $\alpha$ has been determined, we realize that on the branch described above in \eqref{coneq1} and \eqref{coneq2}, everything is indeed higher order to the $\mathcal{O} (c_0)$ term.  Thus, $c_0=0$ lest we move off the mass $1$ branch. {Therefore, we have $\psi_+ = \psi_-$ for sufficiently small $\alpha$.}
\end{proof}

\begin{rem}
    We note that the nature of the kernel of $D_\psi F$ is not so surprising at
    $\alpha = 0$, as a major symmetry of $\calE_{0}$ would be
    {to multiply $(\psi_+, \psi_-)$ by a rotation
      matrix}, which is an invariant of the Lagrangian. However, given
    that at $\alpha = 0$, we have $\psi_+ = \psi_- = \phi$, this
    symmetry generates no new solutions except the one we have found
    in the kernel.  Using the convexity {of $\calE_{0}$},
    we have uniqueness of the symmetric solution $\phi$ as a minimizer
    having fixed mass $\| \phi \|_{L^2} =1$. 
 \end{rem}
    
\begin{rem}
  From the sign changes \eqref{coneq1} and \eqref{coneq2}, we expect
  that with no mass constraint the branch construction stemming from
  the kernel of $D_\psi F$ to leading order leads to $E_+ = -E_-$.  If we
  were allowed to make such a symmetric reduction, the arguments above
  can be simplified.
\end{rem}

\subsection{Construction of the local branch under the symmetry assumption}

In Proposition \ref{prop:sym} above, we established that the
bifurcation off $\alpha =0$ occurs in the symmetry class such that
$\psi_+ =\psi_-$.  Within this symmetry class, we demonstrate in this
section that one may construct a unique local branch of solutions that preserves the mass of the electronic states as $1$.  We could have absorbed this constraint above in a modification of the application of the implicit function theorem, but for simplicity of exposition we have split the two arguments apart.

Using that the linearization preserve the symmetry of solutions proven in Proposition \ref{prop:sym}, let us limit ourselves to solutions of the simplified Euler-Lagrange equation for $\phi (E, \alpha)$ given by {
\begin{equation*}
- \frac12 \Delta \phi(x) + E \phi(x) - V(x) \phi(x) + 2  \int \frac{| \phi |^2 (y) }{|x-y|} dy  \phi(x) - \alpha |\phi(x)|^{\frac23} \phi(x) = 0.
\end{equation*}}
{Denote the mass of $\phi$ by }
\begin{equation*}
M(E,\alpha) := \int | \phi(E, \alpha)|^2 dx.
\end{equation*}
By construction, $M(E_0,0) = 1$.  To find mass $1$ states, using that $\phi = \phi( E, \alpha)$, we wish to find $E(\alpha)$ solving
\begin{equation*}
M(E, \alpha) = \int |\phi (E(\alpha), \alpha)|^2 dx - 1= 0.
\end{equation*}
Hence, we apply the Implicit Function Theorem once more, which guarantees the solvability of $E(\alpha)$ provided 
\begin{equation*}
\frac{\partial M}{\partial E} \Bigr\vert_{E = E_0, \alpha = 0}  \neq 0 .
\end{equation*}
However, at $\alpha = 0$ this follows directly from \eqref{eq:derE0}. 
Using the Implicit Function Theorem for a small range of $\alpha$, there is an $E = E(\alpha)$ satisfying the mass constraint.  Thus the proof of Proposition \ref{prop:smallalpha} is complete.

\section{Localization and symmetry breaking for large $\alpha$}
\label{sec:large}

{In this section, we prove Theorem \ref{thm2} by classifying the large $\alpha$ behavior of the minimizer.}

\subsection{A priori energy estimate}

We consider a variational problem with only the kinetic and exchange
terms:
\begin{equation}\label{eq:semilinear}
  \min_{\varphi: \int \abs{\varphi}^2 = 1} \mc{F}(\varphi) = \frac{1}{2} \int \abs{\nabla \varphi}^2 - \int \abs{\varphi}^{8/3}.
\end{equation}
It is now classical in the theory of nonlinear Schr\"odinger equations that the minimizer of \eqref{eq:semilinear} exists.  {In fact, there is a unique radial minimizer, and all minimizers are translated versions of it, see for instance \cite{SS}.}
Denote $\varphi$ the radial minimizer of \eqref{eq:semilinear} centered at
zero, it satisfies
\begin{equation}
  - \frac{1}{2} \Delta \varphi  - \frac{4}{3} \abs{\varphi}^{2/3} \varphi + E \varphi = 0
\end{equation}
with $E$ being a strictly positive Lagrange multiplier. Moreover,
$\varphi$ decays exponentially as $\abs{x} \to \infty$.

We consider dilation operator {$D_\alpha$ for $\alpha > 0$} that preserves the
$L^2$ norm
{
\begin{equation}
  (D_\alpha f)(x) = \alpha^{3/2} f(\alpha x). 
\end{equation} }
Let $x_+$ and $x_-$ minimize  
\begin{equation}
  \min  \Bigl( \norm{\nabla \psi_{\pm} - \nabla (D_{\alpha} \varphi)(\cdot - x_{\pm})}_{L^2}^2 + E \norm{\psi_{\pm} - (D_{\alpha} \varphi)(\cdot - x_{\pm})}_{L^2}^2 \Bigr)  .
\end{equation}
We write the remainder as 
\begin{equation}\label{eq:remainder}
  \psi_{\pm} =   \bigl(D_{\alpha} (\varphi + w_{\pm}) \bigr)(\cdot - x_{\pm}). 
\end{equation}
Correspondingly, we have $\varphi = D_{\alpha}^{-1}
\tau_{x_{\pm}}^{-1} (\psi_{\pm}) - w_{\pm}$, to simplify notation, we denote 
\begin{equation}
  \wt{\psi}_{\pm} = D_{\alpha}^{-1} \tau_{x_{\pm}}^{-1} \psi_{\pm} 
  = \alpha^{-3/2} \psi_{\pm}\Bigl( \frac{x + x_{\pm}}{\alpha} \Bigr). 
\end{equation}

As $\{\psi_{\pm}\}$ minimize $\calE_{\alpha}$, we have
\begin{align}
  0  & \leq  \calE_{\alpha}\bigl( (D_{\alpha} \varphi)(\cdot - x_+), 
       (D_{\alpha} \varphi)(\cdot - x_-) \bigr) - \calE_{\alpha}(\psi_+, \psi_-) \notag \\
     & = \alpha^2  \bigl( 2 \mc{F}(\varphi) - \mc{F}(\wt{\psi}_+) - \mc{F}(\wt{\psi}_-) \bigr) + \int V (\rho_{\varphi} - \rho_{\psi})  \label{alphainfbd1}  \\
     & \hspace{1cm}+ \frac{1}{2} \iint \frac{\rho_{\varphi}(x) \rho_{\varphi}(y)}{\abs{x - y}} \ud x \ud y - \frac{1}{2} \iint \frac{\rho_{\psi}(x) \rho_{\psi}(y)}{\abs{x - y}} \ud x \ud y,  \notag
\end{align}
where we have set
\begin{equation}
\label{rhovarphieq}
  \rho_{\varphi}(x) = \abs{(D_{\alpha} \varphi)(x - x_+)}^2 + 
  \abs{(D_{\alpha} \varphi)(x - x_-)}^2. 
\end{equation}

Note that $\frac{1}{2} \iint \frac{\rho_{\psi}(x) \rho_{\psi}(y)}{\abs{x - y}} \ud x \ud y \geq 0$ {and $\int V \rho_{\varphi} \leq 0$, rearranging the terms, we obtain 
\begin{equation}\label{eq:scaledF}
  \mc{F}\bigl( \wt{\psi}_+ \bigr) + \mc{F}\bigl( \wt{\psi}_- \bigr) - 2 \mc{F}(\varphi)
  \leq \frac{1}{\alpha^2} \frac{1}{2}\iint \frac{\rho_{\varphi}(x) \rho_{\varphi}(y)}{\abs{x - y}} \ud x \ud y   - \frac{1}{\alpha^2} \int V \rho_{\psi} \ud x. 
\end{equation}
For the first term on the right hand side, we have by the definition of $\rho_{\varphi}$ in \eqref{rhovarphieq} that 
\begin{equation}
  \begin{aligned}
    \frac{1}{2}\iint \frac{\rho_{\varphi}(x)
      \rho_{\varphi}(y)}{\abs{x - y}} & = \frac{1}{2} \iint
    \frac{\bigl(\abs{(D_{\alpha} \varphi)(x - x_+)}^2 +
      \abs{(D_{\alpha} \varphi)(x - x_-)}^2\bigr) \bigl(
      \abs{(D_{\alpha} \varphi)(y - x_+)}^2 +
      \abs{(D_{\alpha} \varphi)(y - x_-)}^2\bigr)}{\abs{x - y}} \\
    & \leq \iint \frac{\abs{(D_{\alpha} \varphi)(x - x_+)}^2  \abs{(D_{\alpha} \varphi)(y - x_+)}^2 }{\abs{x - y}} + \iint \frac{\abs{(D_{\alpha} \varphi)(x - x_-)}^2  \abs{(D_{\alpha} \varphi)(y - x_-)}^2 }{\abs{x - y}} \\
    & = 2 \iint \frac{\abs{(D_{\alpha} \varphi)(x)}^2
      \abs{(D_{\alpha} \varphi)(y)}^2 }{\abs{x - y}} \\
    & = 2 \alpha \iint \frac{\abs{\varphi(x)}^2
      \abs{\varphi(y)}^2 }{\abs{x - y}}, 
\end{aligned}
\end{equation}
where we have used the scaling relation of $D_{\alpha}$ and change of variables $\alpha x \mapsto x, \alpha y \mapsto y$ in the last equality. 
To control the second term on the right hand side  of \eqref{eq:scaledF}}, recall that by Hardy's uncertainty
principle, we have for any $X \in \RR^3$ and $f \in H^1$
\begin{equation}\label{eq:Hardy}
  \int \frac{1}{\abs{x - X}} \abs{f(x)}^2 \ud x  
  \leq 4 \norm{f} \norm{\nabla f}. 
\end{equation}
Therefore, since $\norm{\psi_{\pm}} = 1$, we have 
\begin{equation}
  - \int V \rho_{\psi} \ud x = \int \frac{1}{\abs{x - e_1}} \bigl(
  \abs{\psi_+}^2 + \abs{\psi_-}^2) \ud x + \int \frac{1}{\abs{x +
      e_1}} \bigl( \abs{\psi_+}^2 + \abs{\psi_-}^2) \ud x \\
  \leq C \bigl( \norm{\nabla \psi_-} + \norm{\nabla
    \psi_+} \bigr).
\end{equation}
Thus, we arrive at
\begin{equation}\label{eq:Fupper}
  \begin{aligned}
    \mc{F}\bigl( \wt{\psi}_+ \bigr) + \mc{F}\bigl( \wt{\psi}_- \bigr) - 2
    \mc{F}(\varphi) & \leq \frac{C}{\alpha} + \frac{C}{\alpha^2} \bigl(
    \norm{\nabla \psi_+} + \norm{\nabla \psi_-} \bigr)   \\
    & \leq \frac{C}{\alpha} + \frac{C}{\alpha^2}  \bigl( \norm{\nabla D_{\alpha}(\varphi + w_+) } + \norm{\nabla D_{\alpha}(\varphi + w_-) } \bigr) \\
    & \leq \frac{C}{\alpha} + \frac{C}{\alpha} \bigl( \norm{\nabla
      w_+} + \norm{\nabla w_-} \bigr) \\
    & \leq \frac{C}{\alpha} + \frac{C}{\alpha} \bigl( \norm{\nabla
      w_+}^2 + \norm{\nabla w_-}^2 \bigr).
  \end{aligned}
\end{equation}
Using the result in \cite{Weinstein} for the semilinear functional \eqref{eq:semilinear}, the left hand side of \eqref{eq:Fupper} is bounded from
below as
\begin{equation}\label{eq:Flower}
  \mc{F}\bigl( \wt{\psi}_+ \bigr) + \mc{F}\bigl( \wt{\psi}_- \bigr) - 2
  \mc{F}(\varphi) \geq g\bigl(\norm{w_+}_{H^1}\bigr) + g\bigl( \norm{w_-}_{H^1}\bigr), 
\end{equation}
where 
\begin{equation}
  g(t) = c t^2 ( 1 - a t^{\theta} - b t^4) \quad \text{with} \quad a, b, c, \theta > 0. 
\end{equation}
Combining \eqref{eq:Fupper} and \eqref{eq:Flower}, we conclude that
\begin{equation}
  \lim_{\alpha \to \infty} \norm{w_{\pm}}_{H^1} = 
  \lim_{\alpha \to \infty} \norm{\wt{\psi}_{\pm} - \varphi}_{H^1} = 0. 
\end{equation}
In other words, up to translation and dilation, the minimizer of
\eqref{eq:UDFT} is close to the minimizer of the semilinear problem
\eqref{eq:semilinear} for $\alpha$ large.  {This establishes the $H^1$ convergence stated in Theorem \ref{thm2}.  We now proceed to establish the exact structure of the minimizer as stated in \eqref{eqn:minform}.}

\subsection{Location optimization}

We further determine the translation vectors $x_{\pm}$.  We claim that
as $\alpha \to \infty$, the translation vectors $x_{\pm} \to \pm e_1$
(up to swapping $x_+$ and $x_-$, recall that swapping $\psi_+$ and
$\psi_-$ does not change the energy).  The key observation is that the
kinetic and exchange energy terms are invariant with respect to
translation, and hence $x_{\pm}$ are determined by the potential and
Coulomb repulsion terms, which are higher order terms when $\alpha$ is
large.

For this, we consider shifted minimizers 
\begin{equation}
  \wh{\psi}_+ = \psi_+( \cdot + e_1 + x_+)
  \quad \text{and} \quad \wh{\psi}_- = \psi_-( \cdot - e_1 + x_-). 
\end{equation}
By \eqref{eq:remainder}, we have 
\begin{equation}
  \wh{\psi}_{\pm} = (D_{\alpha}\varphi)( \cdot \pm e_1)
  + (D_{\alpha} w_{\pm})( \cdot \pm e_1).
\end{equation}

Due to minimality, we have
\begin{equation}\label{eq:diffhatpsi}
  \begin{aligned}
    0 & \leq \calE_{\alpha}( \wh{\psi}_+, \wh{\psi}_-) - \calE_{\alpha}(
    \psi_+, \psi_-) \\
    & = \int V(\rho_{\wh{\psi}} - \rho_{\psi}) + \frac{1}{2} \iint
    \frac{\rho_{\wh{\psi}}(x) \rho_{\wh{\psi}}(y)}{\abs{x - y}} \ud x
    \ud y - \frac{1}{2} \iint \frac{\rho_{\psi}(x)
      \rho_{\psi}(y)}{\abs{x - y}} \ud x \ud y.
  \end{aligned}
\end{equation}
Recall $\rho_{\varphi}$ and similarly define $\rho_{\wh{\varphi}}$ as
\begin{align*}
  & \rho_{\varphi}(x) = \abs{(D_{\alpha} \varphi)(x - x_+)}^2 +
  \abs{(D_{\alpha} \varphi)(x - x_-)}^2; \\
  & \rho_{\wh{\varphi}}(x) = \abs{(D_{\alpha} \varphi)(x + e_1)}^2 +
  \abs{(D_{\alpha} \varphi)(x - e_1)}^2. 
\end{align*}
Denoting 
\begin{equation}
  \delta_{\text{VC}}( \rho_1, \rho_2) = \int V(\rho_1 - \rho_2) + \frac{1}{2} \iint
  \frac{\rho_1(x) \rho_1(y)}{\abs{x - y}} \ud x \ud y - \frac{1}{2} \iint \frac{\rho_2(x) \rho_2(y)}{\abs{x - y}} \ud x \ud y, 
\end{equation}
we rewrite \eqref{eq:diffhatpsi} as
\begin{equation}
  \delta_{\text{VC}}( \rho_{\wh{\psi}}, \rho_{\psi}) = \delta_{\text{VC}}( \rho_{\wh{\psi}}, \rho_{\wh{\varphi}}) + \delta_{\text{VC}}(\rho_{\wh{\varphi}}, \rho_{\varphi}) + \delta_{\text{VC}}(\rho_{\varphi}, \rho_{\psi}) \geq 0. 
\end{equation}

Let us estimate $\delta_{\text{VC}}(\rho_{\varphi}, \rho_{\psi})$
first. For the potential term, using \eqref{eq:Hardy} for $f =
\abs{\rho_{\varphi} - \rho_{\psi}}^{1/2}$,
\begin{equation}\label{eq:potentialdiff}
  \int  \abs{V (\rho_{\varphi} - \rho_{\psi})} \leq C 
  \norm{\abs{\rho_{\varphi} - \rho_{\psi}}^{1/2}}_{L^2}\norm{\nabla \abs{\rho_{\varphi} - \rho_{\psi}}^{1/2}}_{L^2} \leq C \norm{\nabla \abs{\rho_{\varphi} - \rho_{\psi}}^{1/2}}_{L^2}. 
\end{equation}
For the difference in Coulomb energy
\begin{multline}
  \left \lvert \frac{1}{2} \iint \frac{\rho_{\varphi}(x)
      \rho_{\varphi}(y)}{\abs{x - y}} \ud x \ud y - \frac{1}{2} \iint
    \frac{\rho_{\psi}(x) \rho_{\psi}(y)}{\abs{x - y}} \ud x
    \ud y \right\rvert \\
  \leq \iint \frac{\abs{\rho_{\varphi} - \rho_{\psi}}(x)
    \rho_{\varphi}(y)}{\abs{x - y}} \ud x \ud y
  + \frac{1}{2} \iint \frac{(\rho_{\varphi} - \rho_{\psi})(x)(\rho_{\varphi} - \rho_{\psi})(y)}{\abs{x - y}} \ud x \ud y \\
  \leq C \norm{\rho_{\varphi} - \rho_{\psi}}_{L^{3/2}}
  \norm{\rho_{\varphi}}_{L^1} + C \norm{\rho_{\varphi} -
    \rho_{\psi}}_{L^{6/5}}^2,
\end{multline}
where the last line uses the Hardy-Littlewood-Sobolev inequality.
Observe that using  interpolation and Gagliardo-Nirenberg-Sobolev
inequality, we have
\begin{align}
  & \norm{f}_{L^{6/5}} \leq \norm{f}_{L^1}^{3/4} \norm{f}_{L^3}^{1/4}
    \leq C \norm{f}_{L^1}^{3/4} \norm{\nabla \sqrt{f}}_{L^2}^{1/2}, \\
  & \norm{f}_{L^{3/2}} \leq \norm{f}_{L^1}^{1/2} \norm{f}_{L^3}^{1/2}
    \leq C \norm{f}_{L^1}^{1/2} \norm{\nabla \sqrt{f}}_{L^2}. 
\end{align}
Combined with the above three inequalities, we get
\begin{equation}\label{eq:coulombdiff}
  \left \lvert \frac{1}{2} \iint \frac{\rho_{\varphi}(x)
      \rho_{\varphi}(y)}{\abs{x - y}} \ud x \ud y - \frac{1}{2} \iint
    \frac{\rho_{\psi}(x) \rho_{\psi}(y)}{\abs{x - y}} \ud x
    \ud y \right\rvert 
  \leq C \norm{\nabla
    \abs{\rho_{\varphi} - \rho_{\psi}}^{1/2}}_{L^2}.
\end{equation}

To estimate the right hand side of \eqref{eq:potentialdiff} and
\eqref{eq:coulombdiff}, by definition 
\begin{equation}
  \begin{aligned}
    \norm{\nabla \abs{\rho_{\varphi} - \rho_{\psi}}^{1/2}}_{L^2} &
    \leq \norm{\nabla \bigl(2 \abs{D_{\alpha} \varphi} \abs{D_{\alpha}
        w_+} + 2 \abs{D_{\alpha} \varphi} \abs{D_{\alpha} w_-} +
      \abs{D_{\alpha} w_+}^2 + \abs{D_{\alpha} w_-}^2
      \bigr)^{1/2}}_{L^2} \\
    & \leq C \biggl( \norm{\nabla \bigl( \abs{D_{\alpha} \varphi}
      \abs{D_{\alpha} w_+}\bigr)^{1/2}}_{L^2} + \norm{\nabla \bigl(
      \abs{D_{\alpha} \varphi} \abs{D_{\alpha} w_-}\bigr)^{1/2}}_{L^2}
    \\
    & \qquad \qquad + \norm{\nabla \abs{D_{\alpha} w_+}}_{L^2} +
    \norm{\nabla \abs{D_{\alpha} w_-}}_{L^2} \biggr) \\
    & \leq C \alpha \bigl(\norm{w_+}_{H^1} + \norm{w_-}_{H^1}\bigr), 
  \end{aligned}
\end{equation}
where we have used the convexity of $\abs{\nabla \sqrt{\rho}}^2$  and the Cauchy-Schwarz inequality.
Note that the $\alpha$ pre-factor on the right hand side is natural
from the scaling, since the characteristic length scale of $\rho_{\varphi}$ is order $1/\alpha$ due to the construction by dilation. Therefore, to sum up,  \begin{equation}
  \abs{\delta_{\text{VC}}(\rho_{\varphi}, \rho_{\psi})} \leq C \alpha \bigl(\norm{w_+}_{H^1} + \norm{w_-}_{H^1}\bigr). 
\end{equation}
It is easy to check that the same upper bound also holds for
$\delta_{\text{VC}}(\rho_{\wh{\varphi}}, \rho_{\wh{\psi}})$. Thus, 
\begin{equation}
  \delta_{\text{VC}}(\rho_{\wh{\varphi}}, \rho_{\varphi}) \geq - C \alpha \bigl(\norm{w_+}_{H^1} + \norm{w_-}_{H^1}\bigr). 
\end{equation}

We now turn the above estimate of the Coulomb energy difference into
an estimate of the translation vectors $x_{\pm}$. For this, we
calculate more explicitly $\delta_{\text{VC}}(\rho_{\wh{\varphi}},
\rho_{\varphi})$ (recall that $\varphi$ is the unique radial minimizer
to the semilinear functional \eqref{eq:semilinear}). 
We have 
\begin{multline}
\delta_{\text{VC}}(\rho_{\varphi}, \rho_{\wh{\varphi}}) = 
\int V ( \rho_{\varphi} - \rho_{\wh{\varphi}}) + \iint \frac{\abs{D_{\alpha} \varphi}^2(x - x_+) \abs{D_{\alpha} \varphi}^2 (y - x_-)}{\abs{x - y}} \ud x \ud y \\
-
\iint \frac{\abs{D_{\alpha} \varphi}^2(x - x_+) \abs{D_{\alpha} \varphi}^2(y - x_-)}{\abs{x - y}} \ud x \ud y .
\end{multline}
As $\varphi$ decays exponentially, we have 
\begin{equation}
\iint \frac{\abs{D_{\alpha} \varphi}^2(x - x_+) \abs{D_{\alpha} \varphi}^2(y - x_-)}{\abs{x - y}} \ud x \ud y \lesssim \frac{1}{\alpha}, 
\end{equation}
and therefore 
\begin{equation}\label{eq:positionxpm}
  \int V ( \rho_{\varphi} - \rho_{\wh{\varphi}}) + \iint \frac{\abs{D_{\alpha} \varphi}^2(x - x_+) \abs{D_{\alpha} \varphi}^2(y - x_-)}{\abs{x - y}} \ud x \ud y \leq    
  C \alpha \bigl(\norm{w_+}_{H^1} + \norm{w_-}_{H^1}\bigr) + \Or(\alpha^{-1}). 
\end{equation}
This implies 
\begin{multline}
  \lim_{\alpha \to \infty} \frac{1}{\alpha} \int V ( \rho_{\varphi} - \rho_{\wh{\varphi}}) \\
= \lim_{\alpha \to \infty} - \frac{1}{\alpha} \int \Bigl(\frac{1}{\abs{x - e_1}} + \frac{1}{\abs{x + e_1}}\Bigr)  \bigl( \abs{D_{\alpha} \varphi}^2(x - x_+) 
- \abs{D_{\alpha} \varphi}^2(x - e_1) \bigr) \ud x \\+
\lim_{\alpha \to \infty}  - \frac{1}{\alpha} \int \Bigl(\frac{1}{\abs{x - e_1}} + \frac{1}{\abs{x + e_1}}\Bigr)  \bigl( \abs{D_{\alpha} \varphi}^2(x - x_-) - \abs{D_{\alpha} \varphi}^2(x + e_1) \bigr) \ud x 
= 0, 
\end{multline}
since the second term on the left hand side of \eqref{eq:positionxpm} is non-negative. Note that the two limits in the middle of the above equation are both non-negative. We have 
\begin{equation}
  \lim_{\alpha \to \infty} - \frac{1}{\alpha} \int \Bigl(\frac{1}{\abs{x - e_1}} + \frac{1}{\abs{x + e_1}}\Bigr)  \bigl( \abs{D_{\alpha} \varphi}^2(x - x_+) 
- \abs{D_{\alpha} \varphi}^2(x - e_1) \bigr) \ud x = 0. 
\end{equation}
This implies that $ \min \{\abs{x_+ - e_1}, \abs{x_+ + e_1} \}$
converges to zero, and similarly for $x_-$. Thus, as $\alpha \to
\infty$, $x_{\pm}$ approaches $\{e_1, -e_1\}$. They cannot converge to
the same point, as otherwise the Coulomb interaction is obviously
higher. Therefore, we arrive at the conclusion of Theorem \ref{thm2}.

\section{Numerical solution to Kohn-Sham SDFT (KS-SDFT) equations for $\psi_+$ and $\psi_-$ as a function of $\alpha$}

\label{sec:numresults}

The variation of the  energy functional given in Eq. (\ref{eq:LDA}) with normalization constraints, leads to Euler-Lagrange equations defining
the spin up, $\psi_{+}$, and the spin down, $\psi_{-}$,  KS-SDFT orbital solutions as a function of exchange strength $\alpha$ {and the internuclear bond length, $2R$}. 
In this section we outline the  finite element methods (FEM)  \cite{braess2001finite, brenner2008mathematical, HoudongThesis} we used to produce numerical solutions to these equations and determine their stability.   
These solutions, characterized by the symmetry of the orbital functions and their localization within the molecular framework, were used to explore the transitions between the regions of stability identified by the theorems in Sections \ref{sec:small} and \ref{sec:large}.  
{In the process of generating numerically stable solutions to the Euler Lagrange equations several new classes of solutions were identified. These may have important consequences for the application of Kohn-Sham methods but were not analyzed in Sections \ref{sec:small} and \ref{sec:large}.
The stability and stationary character of the solutions generated with variation in the $R$ parameter are validated via Hessian analysis (see the Appendix)}.

An important feature of the FEM approach we used is that the
expansions of $\psi_-$ and $\psi_+$ in the FEM basis \cite{fenics2012}
are not constrained by any preconceived notion as to the nature of the
solution as is implicit in the atomic orbital expansion basis of
quantum chemistry software \cite{Parr1972,Szabo1989,gauss}. This is
particularly important in our application because of the form of the
KS solutions to (\ref{eq:LDA}) as a function of $\alpha$ (e.g., for
large $\alpha$) are unknown. The numerical problem and the FEM method
we developed for its solution are described in more detail in the
Appendix. A novel feature of the numerical method we have used is that
its time to solution scales linearly with the size of the basis
\cite{HoudongThesis}.

\subsection{Overview of numerical method (FEM)}
Our numerical implementations are based on application of the Python
FEniCS finite element (FEM) package \cite{fenics2012, HoudongThesis},
which is a collection of free software with an extensive list of
features for automated, efficient, finite element solution methods for
differential equations. The source codes implementing the linear
scaling finite element solver described below can be found at the
FEniCS project homepage.\footnote{\url{https://fenicsproject.org}}
More details specific to our calculation are given in the Appendix.

The FEM calculation domain used here is a fixed square box of
dimension $50\times 50\times50$ atomic units which easily contains the
{\Hm} molecule (size $\approx$ 2 atomic units). Because the bound
state molecular orbitals decay exponentially away from the positions
of the nucleus, we apply zero boundary conditions at the domain edges
for the wave functions. The Coulomb potentials {required in the calculation}
are calculated from Poisson's equation using free space boundary
conditions. The singularities of the attractive nuclear potentials,
(\ref{eq:defVR}) are numerically removed by adding a small
positive constant in the denominator \cite{HoudongThesis}.
 
To accommodate the more rapid variation of the $\psi$ functions near
the atomic nucleus the finite element grid is adapted within the
domain, see Figure~\ref{fig:grid}.  This is an essential feature of
atomic and molecular electronic structure calculations
\cite{bylaska2009adaptive,weare1997adaptive,weare1995parallel} that do
not introduce pseudopotentials \cite{kubicki}.

\begin{figure}[H]
\centering
\includegraphics[scale=.45]{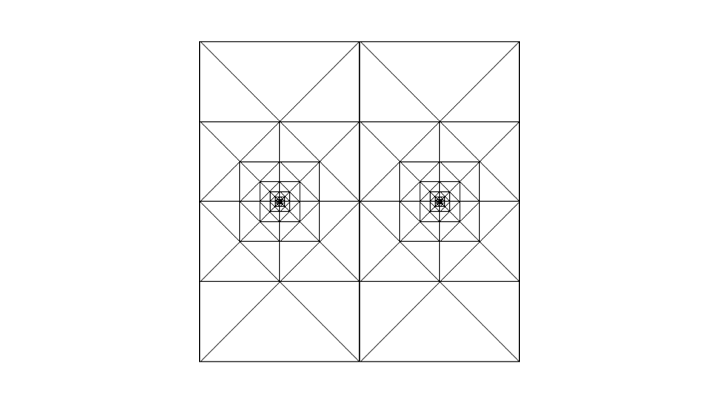}
\caption{A representation of the finite element grid used in calculation. Each triangle represents a tetrahedron in the real calculation. Note that the density of the mesh is significantly increased near the two atomic nuclei (see Appendix for further discussion as to how this mesh is generated).}
\label{fig:grid}
\end{figure}

In these FEM calculations each molecular orbital ($\psi_+$ or $\psi_-$) is written as an expansion in a finite element basis, $\eta_{i}$, with local support centered on the grid points in Figure \ref{fig:grid}\cite{fenics2012,bylaska2009adaptive,braess2001finite}, giving, 
\begin{equation}\label{eq:psi expan}
\psi_{\pm}=\sum_{j=1}^{M} c_{\pm, j} \eta_{j}.
\end{equation}
There are $M$ basis functions, where $M$ is the total number of points in the grid and $\eta_i$ the finite element basis functions (piecewise linear elements with local support) (see Appendix  and \cite{bylaska2009adaptive} for more detail).
\noindent
The variation of the functional (\ref{eq:LDA}) expanded in the basis as in (\ref{eq:psi expan}) leads to generalized eigenvalue problems which must be solved in a self-consistent fashion. 
These may be written as,
\begin{equation}\label{helm}
(-\frac{1}{2} \Delta -\epsilon_{\pm}) \psi_{\pm,i}(x)=V_{\mathrm{eff}, \pm}[\rho_{\pm}] \psi_{\pm}(x),
\end{equation}
{where the spin electron density is  $\rho_{\pm}=\lvert \psi_{\pm}\rvert^2$  and $V_{\mathrm{eff}, \pm}[\rho_{\pm}]$ denotes the effective (spin)-potential corresponding to $\rho_{\pm}$}.   {
Only the lowest spin up and spin down states are occupied and only these states are found in the solution method (Appendix), thus we only need the lowest eigenfunction in (\ref{helm}).}
The eigenvalue problems, (\ref{helm}), are solved using an iterative process in which for step $k$ the $\psi_{\pm}$ on right hand side of (\ref{helm}) and the orbital energies,  $\epsilon_{\pm}^k$, at step $k$ are assigned the values and functionality from the $k-1$ step (see Appendix and \cite{HoudongThesis}).

(\ref{helm}) is solved
using the FEniCS software package (see Appendix and \cite{HoudongThesis} for more detail). 
This package implements a  conjugate gradient solver (generalized minimal residual method, GMRES \cite{gmres})  after preconditioning with an algebraic multigrid
preconditioner  (AMG, BoomerAMG from the Hypre Library \cite{hypre,
briggs2000multigrid,Boomer2002,Tatebe93,golub}). The application of the AMG solver leads to a linear in basis set size to solution time numerical method  \cite{HoudongThesis}. 

{Initial guesses} for the molecular orbitals (MOs) for the FEM
solutions are necessary to start the iteration. Here we used the H atom
Slater Type Orbitals (STO-3G) generated from the NWChem data base
\cite{valiev2010nwchem,apra2020nwchem} to form molecular orbitals for all $\alpha$.
Given two STO-3G functions centered on the atom centers and designated
as $\phi_1$ and $\phi_2$, the initial unnormalized MOs for symmetric delocalized solutions are
$(\phi_1+\phi_2)/2$.  When localized solutions
are expected the initial functions are {taken to be the  STO-3G functions}  $\phi_1$ and $\phi_2$  localized on the different atomic centers, see
\cite{hehre1969self,gauss}. 

When $\alpha$ is very small (weak exchange), 
the final
solutions are always the {paramagnetic} delocalized states that converge to the same spatial dependance for spin up and spin down states (i.e., $\psi_+ = \psi_-$, {where these are the lowest energy solutions for each spin}). For very large $\alpha$ the lowest energy states may be strongly localized (i.e.,  the spin up and spin down single electron states are localized on different atomic centers). These
localized solutions may not be well approximated by the STO-3G initial
functions. However, we have not had problems with convergence of the
method {used and} described in the Appendix.




{
In summary, in our FEM formalism the forms of the spatial parts of the orbital wave functions are completely independent and the symmetry of the total density is not constrained. However, for most of the stationary solutions that we have found the total {electron density retains the symmetry} of the {\Hm} molecule. We have shown above this to be true for the $\rho$ calculated from the lowest energy solutions of (\ref{eq:LDA}) in both the large $\alpha$ and small $\alpha$ limits. However, for the lowest energy product state the symmetry of the spatial parts of the individual spin orbitals may be broken in a way that preserves the symmetry of the total density of the  molecule leading to localization of the electron spin. Additional higher energy numerical solutions have been identified which do not preserve the symmetry of the molecule (see Figure~\ref{alpha_bifur}). In applications of DFT to large molecules or condensed materials this spin localization is interpreted in term of the observed spin states of lattices (or molecules) (e.g. antiferromagnetic states in condensed materials \cite{Cox1992,Roll2004,Peng-Perdew2017}). }

\subsubsection{Bifurcation in the $R$ dimension}

The optimized total energy as the {\Hm} molecular bond, $2R$, is lengthened at fixed $\alpha= 0.93$ (see (\ref{eq:LDA}) 
similar to the value used in the application of SDFT to molecular and
condensed matter problems) is shown in Figure \ref{E_R}. The accuracies of the total energies calculated are within 0.02 au for the $H_2$ molecule in our calculations reported here, see \cite{HoudongThesis}.  Remarkably, 
for a given $\alpha$ and sufficiently small $R$, the independent solutions for
orbital wave functions $\psi_+$ and $\psi_-$ converge to the same
function even under full variation with no symmetry restriction. ({That is to say, there is NO symmetry breaking in the molecular orbitals.}) This
is consistent with the fixed $R$, small $\alpha$ analysis in Section
\ref{sec:small}. 
 In this region,  the  {\em restricted} DFT
(RDFT) solution in which $\psi_+$ and $\psi_-$ are taken to be the same function (double filling)
is the lowest energy solution to the optimization problem posed in
(\ref{eq:LDA}) even when each orbital function is varied
independently without constraint. Similar behavior is observed in the
Hartree or Hartree-Fock model of electronic structure for the two
electron system. These solutions are important because such doubly filled restricted
DFT solutions are widely assumed and used in quantum chemistry applications
\cite{gauss,Szabo1989}. 

As the {\Hm} bond length is extended as illustrated in Figure
\ref{E_R} ($ 2R \gtrapprox 2.45 \text{au}$) the solution bifurcates
creating two two-electron product (singlet determinant
\cite{Szabo1989,Parr-Yang}) solutions. In the lowest energy state
(lower branch, green line) one symmetry broken electron orbital (say
the spin up state) is localized around one site and the other orbital
function state (spin down) is localized around the second nuclear site (see the
green density distribution cartoon in the bottom right Figure \ref{E_R}).  The
product wave function (total density) for this branch leads to a spin
localized density distribution ({\em spin up and spin down electrons
  localized on different atomic sites with total spin zero and preserving the symmetry of the molecule}). This
spin distribution is consistent with an antiferromagnetic state for
the {\Hm} molecule.  Since spin ordered condensed systems are common
targets for DFT prediction, this is an important dimension for
variation in designing DFT representations of such systems
\cite{Cox1992,Roll2004,Peng-Perdew2017}.

The upper energy branch in Figure \ref{E_R} is a continuation of the
restricted solution in which the spin up and spin down orbitals have
the same spatial dependence (no localization, blue density
distribution bottom right Figure \ref{E_R}). We note that this
solution continues as a stationary solution even for large $R$.
{To better illustrate the structure of the solution as
  $R$ goes from the restricted region to the antiferromagnetic region
  we plot the spin-up density weight of $\psi_+$ defined as
\begin{equation}\label{WF weight}
w_+ = \int_{x_1 \geq 0} \abs{\psi_{+}(x)}^2 \ud x,
\end{equation}
where the integration is on the right half domain corresponding to
regions closer to one of the nuclei. This is the proportion of the
mass of $\psi_+$ localized near one of the nuclei (and $w_+ = 1/2$ if
no localization happens).  Here we have identified $\psi_+$ as the
spin function which after the bifurcation is more localized in the
positive $x_1$ region. We note the smooth behavior of the variation of
$w_+$ as the bond length enters the bifurcated region, observed in
Figure \ref{E_R_Bifur}. The plot further confirms the symmetry
bifurcation as $R$ increases. Note that the symmetry of total electron density is preserved in both the upper and lower states as in insert in  Figure \ref{E_R}.}

\begin{figure}[H]
\centering
\includegraphics[height=6.1cm]{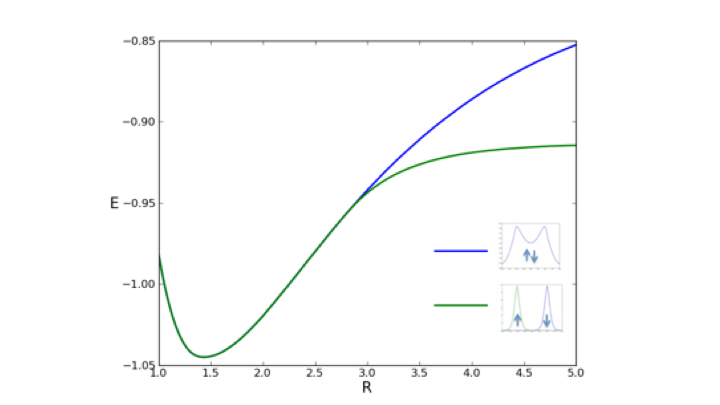}
\caption{Bifurcation of LSDA for {\Hm} in the $R$ dimension with $\alpha = 0.93$. The bifurcation point is in the region $2R=2.40$ to $2.50 \text{au}$. The energy difference between the two states in the region $2R=2.40$ to $2.50 \text{au}$ is of the order of $10^{-5} \text{au}$.}.
\label{E_R}
\end{figure}

\begin{figure}[H]
\centering
\includegraphics[height=6.1cm]{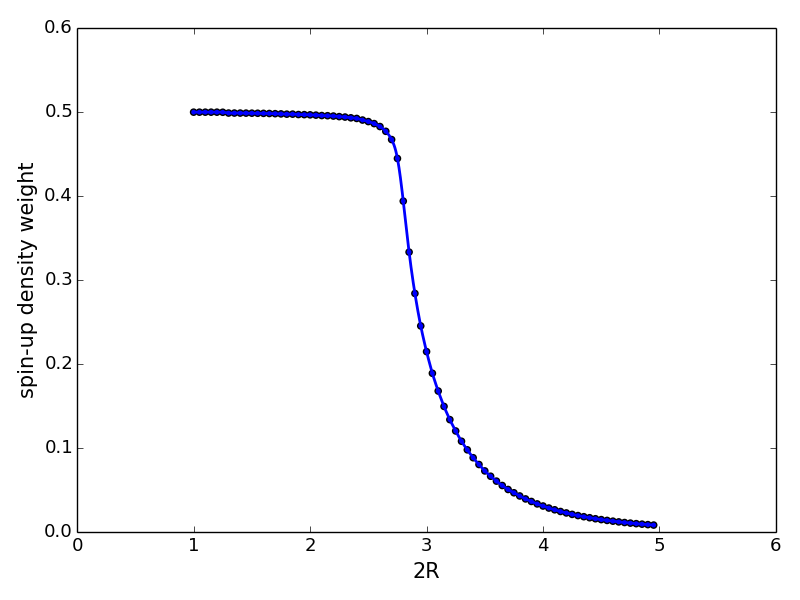}
\caption{Spin-up density weight, $w_+$ (defined in  \eqref{WF weight}) as a function of bond length}
\label{E_R_Bifur}
\end{figure}

\subsubsection{Hessian of bifurcated solutions for R variation}

{After the bifurcation with increasing $R$, Figure \ref{E_R}, 
there are two antisymmetric solutions to the  electron mean field problem  (one spin localized/antiferromagnetic (green line) and one restricted (no localized spin, blue line)).  
This is similar to the application of DFT methods to real systems (e.g., magnetic materials) in which  several solutions (spin orderings) may be found as stationary \cite{kubicki,Roll2004,Peng-Perdew2017}.  These states are frequently interpreted in terms of the relative spin ordering of phases of different structure with apparent reliability. These calculation produce results which correlate well with experimental observations in \cite{kubicki,Roll2004,Peng-Perdew2017}. However, in a realistically sized calculation it can be difficult to identify the minimum energy structure on the basis of currently used optimization methods \cite{kubicki,Roll2004}.  (For a brief overview of how spin is controlled in condensed matter calculations, see reference \cite{kubicki}.)

 The stability/metastability of the solutions to the {\Hm} problem along the 2 branches In Figure (\ref{E_R}) can be determined by analyzing the eigenvalues of the Hessian associated with the optimization problem (\ref{eq:LDA}) (see the Appendix). 
For stationary solutions the gradient of the total energy  (constrained to have the proper normalization) must be zero for any dissent direction.
For stable stationary solutions all eigenvalues of the Hessian (see the Appendix) must be positive. If a solution is unstable there will be at least one negative eigenvalue of the Hessian. 
At the bifurcation point, there will be a zero eigenvalue.}

 Numerical estimates of the eigenvalues of the Hessian the optimization problem (\ref{eq:LDA}) (also calculated via the FEM see appendix) are reported in Table \ref{tab:hess}. These show there is one negative eigenvalue for the delocalized RDFT solution  (green line Figure \ref{E_R}) beyond the bifurcation point. The combination of the zero gradient and the presence of the single negative eigenvalue shows that this is a metastable point in the energy surface.  

\begin{small}
\begin{table}[H]
\begin{tabular}{|c|c|c|c|c|c|}
\hline
$\text{Bond Length}$ & $\text{Solution}$ & $\text{Result}$ & $\text{Details}$ \\
\hline
2.0 a.u. & delocalized& Local Minimizer &  all eigenvalues on the constraint manifold $>0$. \\
\hline
3.5+ a.u. & delocalized  & Saddle Point &  1 eigenvalue on the constraint manifold $<0$. \\
\hline
3.5+ a.u. & localized & Local Minimizer &  all eigenvalues on the constraint manifold $>0$. \\
\hline
\end{tabular}
\caption{Eigenvalues for the Hessian matrix, {$\alpha = 0.93$} for various bond lengths}
\label{tab:hess}
\end{table}
\end{small}
~~

\subsubsection{Bifurcation in the $\alpha$ dimension}~~~

We demonstrate here the numerical verification of the results in Theorems \ref{thm1} and \ref{thm2}.  Adjustment of parameters such as the strength of exchange, $\alpha$, in  \eqref{eq:LDA} in the density functional formalism is sometimes used to improve DFT model performance for spin ordered systems \cite{Henkelman2011,Roll2004}. In the {\Hm}  problem discussed here the parameters $2R$ and $\alpha$ control the bifurcation. For a given $R$ the strength of the exchange term determines the bifurcation point. Figure \ref{alpha_bifur} shows the symmetry breaking bifurcation points for LSDA solutions of (\ref{eq:LDA}) with strength of the exchange contribution for fixed bond lengths $2R=2.0 \text{au}$. 

In the small $\alpha$ setting, there are two identical degenerate spatial solutions (for spin up and spin down). These solutions (delocalized solutions) have peaks at the two atom centers, spread over the whole molecule and have the symmetry of the molecule.  In this region, if numerical solutions are initiated with broken symmetry the $\psi_+$ and $\psi_-$ solutions evolve to have the same spatial dependance,  i.e., $\psi_+$ = $\psi_-$. These solution are equivalent to the single orbital solution of the restricted or doubly filled DFT product function. 

See the analysis in Section \ref{sec:small} for the demonstration of this result, but the underlying reason is that the Coulomb repulsion is somewhat insensitive to the localization of the total density and the kinetic energy dominates over the exchange potential contribution in (\ref{eq:LDA}). Beyond the bifurcation point (as illustrated in Figure \ref{E_R}, the broken symmetry  solutions with excess spin concentrate on each atom (localized  solutions) appears and the product solution {with equivalent spin localization on each site} is the global minimizer. The total density still has the symmetry of the molecule.  The restricted solutions with higher energy are still stationary along the upper branch of the bifurcation curve. These solutions have not been discussed {in our analysis}.

As $\alpha$ is further increased (at constant $R$) a variety of new bifurcations appear. The exchange potential contributes much more than the Coulomb potential so the solution tends to be localized instead of delocalized. 
 We note that the antiferromagnetic solution (blue line) is the global solution for all large $\alpha$. This is an important result since this is the solution generally associated with magnetic behavior in real materials.The symmetric delocalized solution (dark green and light green lines in Figure \ref{alpha_bifur}) is the highest energy. For very high $\alpha$ the maximum density moves to the middle of the bond.  For $\alpha > 6$  the high energy delocalized solutions break symmetry and forms two stable lower energy two electron solutions centered on the atom centers (red line in Figure \ref{alpha_bifur}). 

\begin{figure}[H]
\centering
\includegraphics[height=6.1cm]{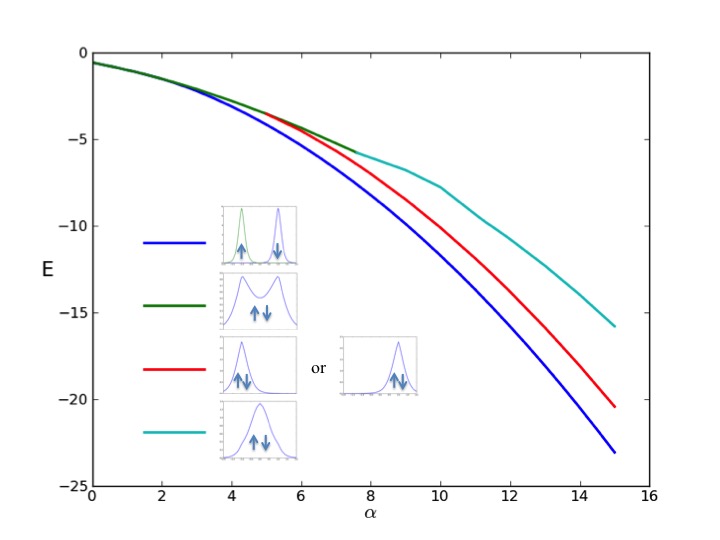} 
\caption{Several numerically constructed branches of the
  bifurcation of LSDA for {\Hm} in the $\alpha$ parameter with $R=2.0 \text{au}$
  showing the relevant $\psi_{\pm}$ profiles.}
\label{alpha_bifur}
\end{figure}

\begin{figure}[H]
\centering
\includegraphics[height=6.1cm]{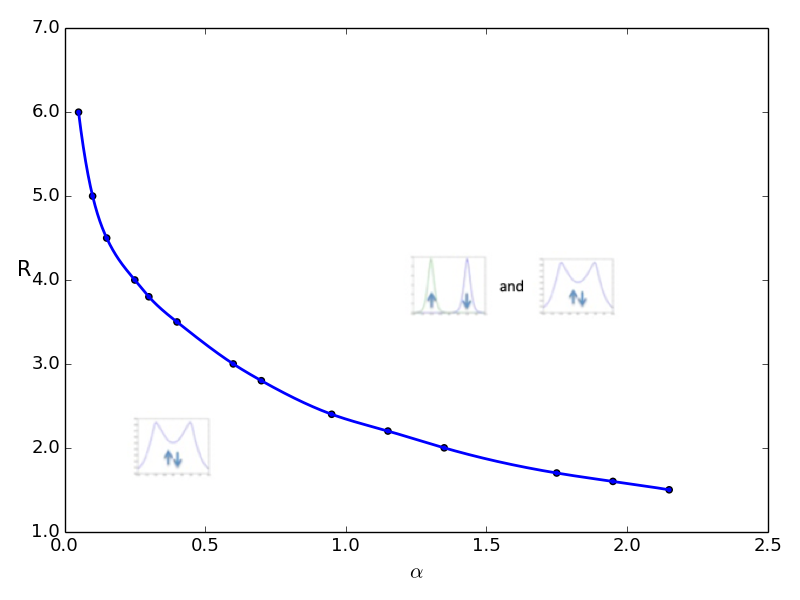}
\caption{The first symmetry
  breaking bifurcation of LSDA for {\Hm} as a phase diagram in $R$
  and $\alpha$. Below the line only delocalized states are present while above the line there are both delocalized and localized states.}
\label{R_alpha}
\end{figure}
{
For applications to real materials the bifurcation in $R$ for a fixed $\alpha$ is taken as the onset of magnetic behavior. The bond length
corresponding to beginning of antiferromagnetic behavior (spin
localization) occurs after the first bifurcation. In Figure \ref{R_alpha} we show the variation of the first $R$ bifurcation point for
different $\alpha$. 
For $\alpha$  in the region commonly used the bond length for bifurcation is quite sensitive to the strength of exchange.}

\section{Discussion and future works}\label{sec:conclusion}

A similar analysis to the large $\alpha$ case gives that the ground state Euler-Lagrange equation for large $R$ can be transformed to an equation of the form 
{ 
\begin{equation}
\label{e:largeR}
-\Delta \psi_{\pm}(x)  + R^2 E (R)\psi_{\pm}(x)  - R V_1(x) \psi_{\pm}(x) + R \int \frac{|\psi_+|^2 (y)+ | \psi_-|^2 (y)}{|x-y|} dy \psi_\pm (x) + R \alpha | \psi_\pm(x)|^{\frac23} \psi_\pm(x) = 0,
\end{equation}
taking $\psi_\pm  (x) = R^{\frac32} \psi_\pm (Rx) $. For $R \gg 1$, this is related to a new problem with large Coulomb repulsion,} large but unit distance apart nuclear masses $Z = R$, and a strong exchange-correlation nonlinearity $R \alpha$.  Thus, the main issue is to study the nature of the stable curve for a large nuclear mass with strong exchange-correlation nonlinearity and observe what the nature of the Lagrange multiplier $R^2 E (R)$ should be as $R \to \infty$.  The intuition is that this scales the problem to be localized since moving along the stable branch of states from low electron mass (small Lagrange multiplier) for the potential $V_1$ to large electron mass (large Lagrange multiplier) eventually concentrates onto localized states over each well.  {This suggests that we consider a modified  Lagrangian with critical points given by
\begin{equation}
\label{e:largeR1}
-\epsilon^2 \Delta \psi_{\pm} (x) + \psi_{\pm} - V_1 \psi_{\pm} + \int \frac{|\psi_+|^2 (y) + | \psi_-|^2 (y) }{|x-y|} dy \psi_\pm (x) +  \alpha | \psi_\pm (x)|^{\frac23} \psi_\pm (x) = 0,
\end{equation}
where the small parameter $\epsilon = 1 / \sqrt{R}$. This looks like a
Ginzburg-Landau type singular-perturbation.} As a result, this
motivates the following question for a (strange) Hydrogen model
{
  (it is strange since a Coulomb self-repulsion and an
  exchange energy for a single electron are included)}: Is the
minimizer of
\begin{equation}
\label{eq:Hydweird}
E_{H}(u) = \frac{1}{2} \int \abs{\nabla
 u}^2 \ud x - \int \frac{Z}{|x|} |u|^2 (x) \ud x 
+ \frac{1}{2} \iint \frac{|u|^2 (x) |u|^2 (y)}{\abs{x-y}} \ud x \ud y - \int |u|^{8/3} dx
\end{equation}
such that $\| u \|_{L^2} = 1$
orbitally stable.  This has been answered in some sense when $Z = 0$ in \cite{ruiz2010schrodinger} when the mass is that of the absolute minimizer. Understanding what occurs for the natural electronic mass $1$ requires further investigation of this model and will be a topic of future work.

\subsection*{Acknowledgements}

We thank the anonymous referee for a very careful reading of this work and making many useful suggestions both to improve the exposition and clarify several points in the analysis.  The work of J.L. is partially supported by the National
  Science Foundation under grants DMS-1312659 and DMS-1454939 and by
  the Alfred P. Sloan Foundation. J.L.M. was supported in part by NSF
  Applied Math Grant DMS-1312874 and NSF CAREER Grant DMS-1352353.

\appendix 

\section{Appendix}~~

\label{sec:numerics}

\begin{appendix}
We recall here the basics of the finite element methods we use in this
work to numerically find the critical points of the XLDA Lagrangian.  The
numerical algorithms are implemented using a {\it python}
implementation of the FENICS finite element package,
\cite{dupont2003fenics}.  Many of the tools we use here are discussed in more detail in
the references \cite{bylaska2009adaptive,HoudongThesis}.  For complete discussions of Finite
Element Methods, see the books of
\cite{axelsson2001finite,bank1981optimal,braess2001finite,brenner2008mathematical}.  The method we develop here takes advantage of the sparsity of the FEM representation of the eigenvalue problem leading to an algorithm that scales linearly with number of basis functions. For
resources on large scale computing in computational chemistry, see
\cite{kendall2000high,valiev2010nwchem}.

\subsection*{I. The Finite Element Set Up} ~~

\begin{figure}[H]
\centering
\includegraphics[scale=.85]{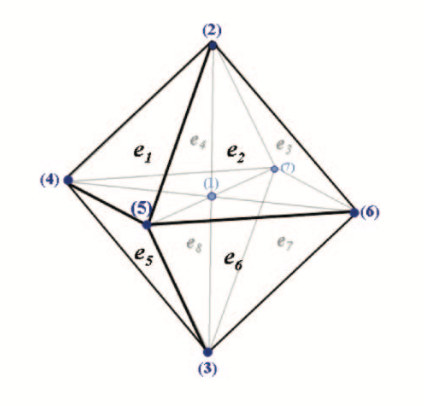}
\caption{Finite element tetrahedron defining FEM elements and
  nodes. The $L$ tetrahedral elements are identified by $e_l$. The
  global node are indexes by $(m)$. For each element $e_l$ a there is
  global node at each corner. Generally local nodes belonging to
  individual tetrahedra are also defined.  See
  \cite{brenner2008mathematical,axelsson2001finite,braess2001finite,bylaska2009adaptive}.
  These are not shown and are managed transparently by the Dolfin
  software \cite{dolfin}. }
\label{fig:tedrahedron.png}
\end{figure}

We assume that the solution, $\psi_{\pm}$   exists in a bounded domain $\Omega \in \mathscr{R}^3$ that can be divided into a set of $L$ non overlapping tetrahedral elements, $\{e_l\}_{l=1}^L$ \cite{brenner2008mathematical,axelsson2001finite,braess2001finite,bylaska2009adaptive}, see Figures~\ref{fig:tedrahedron.png} and \ref{fig:grid}. 
For the electronic structure problems we are concerned with  the atomic potentials represented by  $V(x)$ in the Hamiltonian below, (\ref{eff H FEM}), are singular. This leads to rapid variation of the solution to the eigenvalue problem in this region. In order to obtain accuracy the  FEM grid in this region must have a finer resolution as illustrated in Figure \ref{fig:grid} and discussed in \cite{weare1995parallel,weare1997adaptive,bylaska2009adaptive}.

To construct the grid used in the calculation, we:

1. Use BoxMesh \cite{fenics2012}  to generate a coarse mesh in a $50\times50\times50$ domain. The initial number of cells in each direction is $2$. So the total number of tetrahedra will be $48$ and the total number of vertices will be $27$ in the coarse mesh.

2. Find the closest mesh grids to the nuclei and set the parameter {\em cell\_marker} \cite{fenics2012} 
 true {that tells the code to refine the mesh. If {\em cell\_marker}  = false, it means this mesh will NOT be refined}.

3. Refine the whole grid for $3$ cycles.

\noindent
Generally FEM nodes are located at corners, along boundaries or in the centers of tetrahedral regions \cite{brenner2008mathematical,axelsson2001finite,braess2001finite,githubFEM}. For the calculations here the nodes are located only at the corners of the tetrahedra. These nodes are shared by adjacent tetrahedra as in Figure \ref{fig:tedrahedron.png}. 
Each tetrahedron $l$ has four corner nodes. A global index identifies a node as in Figure \ref{fig:tedrahedron.png} (global node numbers in brackets). There are $M$ global nodes in the construction. In actual calculations a local node index identifying a corner global node with a basis function inside a  particular tetrahedral is also defined in Dolfin \cite{githubFEM,axelsson2001finite,braess2001finite,brenner2008mathematical,bylaska2009adaptive} to identify variation associated with a node within a particular tetrahedron \cite{githubFEM,axelsson2001finite,braess2001finite,brenner2008mathematical}. The somewhat difficult book keeping problem of keeping track of the global variation of the basis functions consistent with their local behavior is taken care of nicely in the FEniCS software, see \cite{fenics2012}.
 
For each node (with global index m and local index i) in a tetrahedral element, $l$, a finite element basis functions $\{\chi^{e_l}_i\}$ is defined. In these calculations the $\{\chi^{e_l}_i\}$ are linear functions centered on local nodes $i$ in element $e_l$ \cite{githubFEM,bylaska2009adaptive,brenner2008mathematical,axelsson2001finite,braess2001finite}. For each global node $m$ the linear basis function $\{\chi^{e_l}_i\}$ is 1 on global node $i$ and zero on all other nodes contained in the tetrahedral elements containing global node $m$. For a particular tetrahedron the linear basis associated with local node $i$ of tetrahedral $e_l$, $\{\chi^{e_l}_i\}$, has value only in tetrahedra $e_l$. Illustrations of how this works are given in \cite{githubFEM}. The local node functions $\{\chi^{e_l}_i\}$ can be assembled in functions centered around the global nodes with index $m$ as the global basis functions $\eta_m(x)$.

A piecewise continuous function (here the approximated
$\psi_{\pm}(x)$) can now be expanded as in \cite{bylaska2009adaptive},
\begin{equation}
  \psi_{\pm}(x)=\sum\limits_{m=1}^M c_{\pm,m}\eta_{m}(x).
\end{equation}
Here, $M$ is the dimension of space of global nodes and $c_{\pm,m}$ is the coefficient of basis element $\eta_m$. The value of the $\psi_{\pm}$ on node $m$ is $c_{\pm,m}$. 
 
\subsection*{II. The Generalized Eigenvalue Problem:}

With the above formulation, solving the Kohn-Sham minimization problem
related to (\ref{eq:LDA}) leads to the generalized eigenvalue
problem ({in many of the following equations the $\pm$
  ($+$ spin up, $-$ spin down) notation has been suppressed to keep
  notation simple}),
\begin{equation}\label{eq:Hc}
\mathbf{H}\mathbf{c}=\mathbf{\epsilon}\mathbf{S}\mathbf{c},
\end{equation}
or 
\begin{equation}
\sum_n H_{mn}c_{n,k}=\epsilon_k \sum_n S_{mn}c_{n,k}
\end{equation}
where $k$ identifies the $k^{th}$ eigenfunction,
\begin{equation}\label{eq:Smn}
S_{mn}=\int\limits_\Omega dx \eta_m(x)\eta_n(x),
\end{equation}  
and
\begin{equation}\label{eq:Hmn}
H_{mn}=\frac{1}{2}\int\limits_\Omega dx\nabla \eta_m(x) \nabla\eta_n(x) +\int\limits_\Omega dx \eta_m(x) V_{\text{eff}} \eta_n(x)
\end {equation}
\noindent
with $V_{\text{eff}}$ given by
\noindent
\begin{equation}\label{eff H FEM}
{V_{\rm eff}=V_{\rm ext}(x)+V_{\rm ee}(\rho)+V_{\rm ex}(\rho,\alpha)= V_{\rm ext}(x)+\int \frac{\rho(x')}{|x-x'|} dx'+V_{\rm ex}(\rho_{\pm},\alpha).}
\end{equation}
\noindent
Here $\rho(x)$ is the total electron density, {and $\rho_\pm$ is the spin density, and $V_{\mathrm{ex}, \pm}(\rho_{\pm},\alpha)$ is given by the scaled Dirac form  }
\begin{equation}\label{eq:mod Dir}
V_{{\rm ex}, \pm}(\rho,\alpha)=\alpha \rho_{\pm}^{1/3}.
\end{equation}
{
  Note that in the spin DFT, the exchange potential depends on the spin component, and thus the effective Hamiltonian for the spin-up and spin-down orbitals are different; while the structure of the problem is the same, and hence we keep the notation (e.g., for $\mathbf{H}$ and $\mathbf{S}$) independent of spin component.}
The matrix $H_{mn}$, (\ref{eq:Hmn}), the overlap matrix $S_{mn}$,
(\ref{eq:Smn}) and integrals over $V(x)$ in (\ref{eff H FEM})
can be obtained from the FEniCS software \cite{fenics2012}. The
calculation of these matrix are also carefully discussed for the
electronic structure problem in \cite{bylaska2009adaptive} and in
general in
\cite{brenner2008mathematical,axelsson2001finite,braess2001finite}.
The full potential $V_{{\rm eff}}$ given by (\ref{eff H
  FEM}) is a function of the density requiring that the eigenvalue
problem, (\ref{eq:Hc}), be solved iteratively until self
consistency is achieved. {We are only interested in the lowest $\pm$ energy solutions, though the methods can be modified to higher energy states as well.}

\subsection*{III. Solution to the generalized eigenvalue problem and the associate Coulomb problem}~~

The objective of the calculation is the solution of the generalized eigenvalue problem,  (\ref{eq:Hc}). However, this requires the input of a current estimate of the Classical potential $V_{\rm ee}$ required in $V_{\text{eff}}$, (\ref{eff H FEM}). This may be found as the solution to the Poisson equation
\begin{equation}\label{eq.poi}
\Delta V_{\rm ee} =-4\pi\rho=-4\pi
\left[\abs{\psi_{+}}^2+ \abs{\psi_{-}}^2\right]
\end{equation}
To solve this PDE $V_{\rm ee}$ is also expanded in the finite element basis as,
\begin{equation}
V_{\rm ee}(x)=\sum\limits_{m=1}^Mv_m\eta_m(x).
\end{equation}
\noindent
(\ref{eq.poi}) is then represented by a system of linear equations giving the $\{v_m\}$.
Given a solution to the Coulomb problem, (\ref{eq.poi}), based on a current density for the iteration, the generalized eigenvalue problem, (\ref{eq:Hc}) is also solved in the $\{\eta_m\}_{m=1}^{M}$  finite element basis functions. Each molecular orbital is represented as
\begin{equation}
\psi_{\pm}=\sum_{\alpha=1}^{M} c_{\pm,m}\eta_{m}.
\end{equation}
Note, for this problem there is only one filled molecular orbital for each spin.
The finite element discretization of the one-electron equation for the
current iteration is given as, for $i = \pm$
\begin{equation}
\label{eq:helm}
(T_i-\epsilon_{i})c_{i}=v_{i}.
\end{equation}
$c_{i}=\{c_{i1},\ldots,c_{im}\}$ are the coefficients of molecular orbital in the expansions of finite element basis. The elements of the operator $(T_i-\epsilon_{i})$, are,
\begin{equation}\label{eq:helm1}
(T_i-\epsilon_{i})_{mn}=\int_{\Omega}\left\{\frac{1}{2}\nabla\eta_{m}\nabla\eta_{n}-\epsilon_{i}\eta_{m}\eta_{n}\right\} d x.
\end{equation}
 The elements of $v_{i}$ are given by 
\begin{equation}
(v_{i})_{m}=\int \eta_m(x) V_{\text{eff}, i}(x)\psi_{i}(x) dx
\end{equation}
are calculated in an iterative process in which $V_{\text{eff}, i}(x)$ is defined for step $k$ from the results of the self-consistent solver in the prior iteration. 

\subsubsection*{Details of the FEniCS solver}~~

The AMG solver is based on a V-cycle \cite{briggs2000multigrid,HoudongThesis}  with a maximum number of multi grid
levels of 25. For each fine to course grid transfer a single pre
smoothing step is taken. For each course to fine transfer a single post smoothing step is taken. These smoothing steps use a symmetric-SOR/Jacobi method. On the coarsest level the
course FEM equation is relaxed by Gaussian elimination. 
In the iteration an energy correction step is applied to update new eigenvalues
after the Helmholtz equation is solved for a set of $\epsilon^{(k)}$
from the prior AMGCG cycle. 
The self-consistent solver converges when the
total energy difference in two consecutive iterations is smaller than a selected tolerance. These are solved via the FEniCs code using the GMRES \cite{gmres} and BOOMER AMG (Algebraic Multigrid \cite{briggs2000multigrid}) packages. The solution to this problem is of order $M$ \cite{HoudongThesis}. 
To improve the convergence of the solution a preconditioning based on the algebraic multigrid method is used  \cite{Tatebe93,Boomer2002}. For an introduction to multigrid methods and their application to problems in electronic structure, see \cite{weare1995parallel,weare1997adaptive,briggs2000multigrid,bramble1993multigrid,bramble1987new,brandt1986algebraic,brandt1985algebraic,hackbusch1985multi,harrison2004multiresolution,mccormick1987multigrid}.

\subsection*{IV. The Wavefunction and Orbital Energy update and iteration}~~
\subsubsection*{IV.i: Some Preliminaries}

The eigenvalue problem, (\ref{eq:helm}), is solved using an iterative process in which for step $k$ the $\psi_{\pm}$ on right hand side of (\ref{eq:helm}) and the orbital energies,  $\epsilon_{\pm}^{(k)}$, at step $k$ are assigned the values and functionality from the $k-1$ step. 
The iteration is developed with the intention of producing linear scaling in the number of FEM basis functions, $M$.  
This is achieved by developing a solver strategy that emphasizes the use of the operator $[(\nabla^2-k^2)]^{-1}$ which is efficiently implemented in multigrid schemes.

The density functional equations  leading to the values of $\epsilon_i^k$ and $\psi_i^k$ for the $k$ values (update from $\psi_i^{(k-1)}$ to $\psi_i^{k}$) are written as

\begin{equation}\label{eq:HF}
{\left[-\frac{1}{2}\nabla_{x}^2 -\epsilon_i^{(k-1)} \right]\psi_i^{k}(x)= \left[V_{\rm ext}^{(k-1)}(x) + V^{(k-1)}_{\rm ee}(\rho^{(k-1)}(x)) + V_{\rm ex}^{(k-1)}(\rho^{(k-1)}(x) ) \right] \psi_i^{(k-1)}(x)}
\end{equation}
where $V_{\rm ext}$ is the external potential from (\ref{eq:defVR}).
The  electron-electron Coulomb potential (calculated from FeniCS as above) is given by,
\begin{equation}\label{eq.poi k-1}
\nabla^2V^{(k-1)}_{\rm ee}(x)=-4\pi\rho^{(k-1)}(x)=-4\pi
\left[\abs{\psi^{(k-1)}_{+}}^2+ \abs{\psi^{(k-1)}_{-}}^2\right].
\end{equation}
The exchange potential is given by,
\begin{equation}\label{eq:Exc}
V^{(k-1)}_{\rm ex}(x)=\alpha \rho^{(k-1)}(x)^{\frac{1}{3}} .
\end{equation}


\noindent
This is now a linear PDE of the form,
\begin{equation}\label{eq:ite dif}
\left[-\frac{1}{2}\nabla^2 -\epsilon^{(k-1)}_i \right]\psi^k_i(x) =f^{(k-1)}_i(x).
\end{equation}
\noindent
Note that all the potential terms in (\ref{eq:HF}) have been collected 
in the function $f^{(k-1)}_i$. We calculate the solution to (\ref{eq:ite dif}) using an efficient multigrid method. Because of the complexity of the grid we use the  AMGCG implemented in the FEniCS software \cite{fenics2012}.


\subsubsection*{IV.ii: Update of the wavefunction, from $\psi^{(k-1)}(x)$ to $\psi^{k}(x)$:}

To initiate the $k^{\rm th}$ iteration ($\psi_i^{(k-1)}$ to $\psi_i^k$) we assume we have solutions $\psi_i^{(k-1)}(x)$ and $\epsilon_i^{(k-1)}$. 
The update of the wavefunction proceeds directly from (\ref{eq:ite dif}) as
\begin{equation}\label{:eq Fock Prop 2}
\psi^{k}_i(x) =\left[-\frac{1}{2}\nabla_{x}^2 -\epsilon^{(k-1)}_i \right] ^{-1}f^{(k-1)}_i(x).
\end{equation}
\noindent
All the functions in this equation are defined from the solution that we obtain from AMGCG.

To complete the iteration cycle we also need an update of the orbital energy $\epsilon_i^{(k-1)}$ to $\epsilon_i^{k}$.
\bigskip

\subsubsection*{IV.iii: Update of the orbital energy:}~~~

  We assume we have $\psi_i^{(k-1)}(x)$ and $\epsilon_i^{(k-1)}$. 
and begin by defining two Greens functions:

The {$(k-1)^{\rm th}$} Green's  function, $G^{(k-1)}_i$,  with energy $\epsilon^{(k-1)}$,
\begin{equation}\label{eq:G(k-1)}
G^{(k-1)}_i=\biggl\{-\frac{1}{2}\nabla^2-\epsilon_i^{(k-1)}\biggr\}^{-1}
\end{equation}
\noindent
and a Green's function, $G^{\rm con}_i$ with the converged DFT orbital energy (from converged solution to DFT equations), $\epsilon_i^{\rm con}$. This is given by,
\begin{equation}\label{gcon}
G^{\rm con}_i=\biggl\{-\frac{1}{2}\nabla^2-\epsilon^{\rm con}_i\biggr\}^{-1}.
\end{equation}


 In the iteration, the updated $\epsilon^{k}_i$, given by
 \begin{equation}\label{updatE}
 \epsilon_i^{k}=\epsilon_i^{(k-1)} +\delta\epsilon_i^{k}.
 \end{equation}
 is taken to be a good approximation to $\epsilon^{\rm con}$.
Using this in (\ref{gcon}) we have,
\begin{equation}\label{eq;G(full e)}
G^{\rm con}_i=\biggl\{-\frac{1}{2}\nabla^2-\epsilon^{\rm con}_i\biggr\}^{-1}=\biggl\{\frac{1}{2}\nabla^2-(\epsilon_i^{(k-1)} + \delta\epsilon^{k}_i)\biggr\}^{-1}
\end{equation}
\noindent
the objective is to calculate an orbital energy correction from these equations using $\psi_i^{(k-1)}$. 

The function $\psi^{\rm con}_i$ satisfies the  orbital PDE, 
\begin{equation}\label{Eq: Fock orbit eq}
\psi_i^{\rm con}(x)=\biggl\{-\frac{1}{2}\nabla^2-\epsilon^{\rm con}_i\biggr\}^{-1}f^{\rm con}_i(x).
\end{equation}
In this equation $\epsilon_i^{\rm con}$ is the converged orbital energy. 

\bigskip
\noindent
We assume that $\psi_i^{(k-1)}$ is a good approximation to $\psi^{\rm con}_i$, i.e., that it approximately satisfies
\begin{equation}\label{eq:full phi(k-1)}
\psi^{(k-1)}_i(x)=\biggl\{-\frac{1}{2}\nabla^2-\epsilon^{\rm con}_i\biggr\}^{-1}f_i^{(k-1)}(x)=\biggl\{-\frac{1}{2}\nabla^2-(\epsilon_i^{(k-1)} + \delta\epsilon^{k}_i) \biggr\}^{-1}f_i^{(k-1)}(x).
\end{equation}

\noindent
Now we  expand the full Green's function (RHS) in the energy correction $\delta\epsilon^{k}_i$ to obtain an equation that will update the orbital energy (find a correction to $\epsilon^{(k-1)}_i$).

We use the operator identity,
\begin{equation}\label{eq:ex G func2}
\frac{1}{(1+a+b)}=\frac{1}{(1+a)} + \frac{1}{(1+a)}\frac{-b}{(1+a+b)},
\end{equation}
to obtain
\begin{equation}\label{eq:prop ex 3}
\begin{split}
\biggl\{-\frac{1}{2}\nabla^2 -\epsilon^{(k-1)}_i -\delta\epsilon_i^{k}\biggr\} ^{-1}=&\biggl\{-\frac{1}{2}\nabla^2 -\epsilon^{(k-1)}_i \biggr\} ^{-1}-\biggl[\biggl\{-\frac{1}{2}\nabla^2 -\epsilon^{(k-1)}_i \biggr\} ^{-1}    \\
&\biggl\{-\delta \epsilon_i^{k}\biggr\}\biggl\{-\frac{1}{2}\nabla^2 -\epsilon^{(k-1)}_i -\delta\epsilon_i^{k}\biggr\} ^{-1}\biggr].
\end{split}
\end{equation}
Iteration of this equation leads to an expression for the propagator to 1st order in $\delta\epsilon^{k}_i$ as,
\begin{equation}\label{eq:prop ex 4}
\begin{split}
\biggl\{-\frac{1}{2}\nabla^2 -\epsilon^{(k-1)}_i -\delta\epsilon_i^{k}\biggr\} ^{-1}=&\biggl\{-\frac{1}{2}\nabla^2 -\epsilon^{(k-1)}_i \biggr\} ^{-1}-\biggl[\biggl\{-\frac{1}{2}\nabla^2 -\epsilon^{(k-1)}_i \biggr\} ^{-1}    \\
& \times \biggl\{\delta \epsilon_i^{k}\biggr\}\biggl\{-\frac{1}{2}\nabla^2 -\epsilon^{(k-1)}_i \biggr\} ^{-1}\biggr].
\end{split}
\end{equation}
We can use this result in (\ref{eq:full phi(k-1)}) to give,

 \begin{equation}\label{eq:lin ex of phi}
 \begin{split}
 \psi^{(k-1)}_i(x)=&\biggl\{-\frac{1}{2}\nabla^2-\epsilon_i^{(k-1)}\biggr\}^{-1}f^{(k-1)}_i(x) \\
 &-\biggl\{-\frac{1}{2}\nabla^2-\epsilon_i^{(k-1)}\biggr\}^{-1}\delta\epsilon_i^{k}\biggl\{-\frac{1}{2}\nabla^2-\epsilon_i^{(k-1)}\biggr\}^{-1}f^{(k-1)}_i(x).
 \end{split}
 \end{equation}
 
 \noindent
 This is more conveniently written in vector notation 
 \cite{claude91} as,  
  \begin{equation}\label{eq:lin ex of phi1}
 \begin{split}
\Ket{\psi^{(k-1)}_i}=&\biggl\{-\frac{1}{2}\nabla^2-\epsilon_i^{(k-1)}\biggr\}^{-1} \Ket{f^{(k-1)}_i}    \\
 &-\biggl\{-\frac{1}{2}\nabla^2-\epsilon_i^{(k-1)}\biggr\}^{-1}\delta\epsilon^{k}_i\biggl\{-\frac{1}{2}\nabla^2-\epsilon_i^{(k-1)}\biggr\}^{-1}\ket{f^{(k-1)}_i}
 \end{split}
 \end{equation}
 or
 \begin{equation}\label{eq:lin ex of phi2}
 \begin{split}
0=&-\Ket{\psi^{(k-1)}_i}+\biggl\{-\frac{1}{2}\nabla^2-\epsilon_i^{(k-1)}\biggr\}^{-1} \ket{f^{(k-1)}_i}    \\
 &-\biggl\{-\frac{1}{2}\nabla^2-\epsilon_i^{(k-1)}\biggr\}^{-1}\delta\epsilon^{k}_i\biggl\{-\frac{1}{2}\nabla^2-\epsilon_i^{(k-1)}\biggr\}^{-1}\ket{f^{(k-1)}_i}.
 \end{split}
 \end{equation}
 \noindent
Closing this equation on the left with $\Bra{f^{(k-1)}}$ gives a linear expression for $\delta\epsilon^{k}_i$ which may be in terms of the $\psi^{k}_i$, (\ref{:eq Fock Prop 2}) , as,
 \begin{equation}\label{eq:energy update1}
 \begin{split}
0=&-\Braket{f_i^{(k-1)}|\psi^{(k-1)}_i}+\Braket{f_i^{(k-1)}|\psi_i^{k}}    \\
&-\delta\epsilon^{k}_i \Braket{\psi^{k}_i|\psi^{k}_i}.
 \end{split}
 \end{equation} 
 This may be solved for $\delta\epsilon^{k}_i$ to obtain
 \begin{equation}\label{eq:energy update2}
 \delta\epsilon^{k}_i=\frac{-\Braket{f^{(k-1)}_i|\psi^{(k-1)}_i}+\Braket{f_i^{(k-1)}|\psi^{k}_i}}{\Braket{\psi^{k}_i|\psi^{k}_i} }.
 \end{equation}
 This gives the update to $\epsilon_i^{(k-1)}$  via (\ref{updatE}) to complete the $k^{\text{th}}$ iteration.

\subsection*{V. The Self Consistent Iteration}~~

Algorithm~\ref{algo:4.1}  summarizes the process followed by the self-consistent solver. An initial guess $(c_{i}^{0},\epsilon_{i}^{0}), i=1,\ldots,n$ is given to start the self-consistent iterations. The solver  stops when the total energy difference in two consecutive iterations is smaller than the tolerance TOL. 

\begin{algorithm}[H]
\caption{The Self-consistent Iteration}
\label{algo:4.1}
\begin{algorithmic}                    
    \STATE \textbf{Input} $(c_{i}^{0},\epsilon_{i}^{0}), i=1,\ldots,n$, TOL;
    \WHILE{$\|\epsilon_{\rm total}^{k}-\epsilon_{\rm total}^{(k-1)}\|>$TOL}
	\STATE Evaluate potentials $V^{k}_{ij}, i,j=1,\ldots,n$ ;
        \STATE Evaluate $v^{k}_{i}, i=1,\ldots,n$ ;
	\STATE Solve the Helmholtz equation, and get updated $\{c_{i}^{(k+1)}, i=1,\ldots,n\}$;
	\STATE energy correction step, and get updated $\{\epsilon_{i}^{(k+1)}, i=1,\ldots,n\}$;
	\STATE $k$++;
    \ENDWHILE
    \STATE \textbf{Output} $(c_{i},\epsilon_{i}), i=1,\ldots,n$.
\end{algorithmic}
\end{algorithm}

\subsection*{VI. Hessian analysis}~~

The model we investigate in this work is the local spin density approximation  (LDA) (without correlation energy contributions).  As above the ground singlet state spin unrestricted density functional theory for this two-electron system defines two orbital wave-functions $(\psi_+,\psi_-)$.  The total energy functional is $E(\psi)$ is 

\begin{multline}
  E_{\alpha}(\psi_+, \psi_-) = \frac{1}{2} \int \abs{\nabla \psi_+}^2
  \ud x + \frac{1}{2} \int \abs{\nabla \psi_-}^2 \ud x
  +   \int V_R(x)\left(\abs{\psi_+(x)}^2+\abs{\psi_-(x)}^2\right)  \ud x \\
  + \frac{1}{2} \iint \cfrac{\left(\abs{\psi_+(x)}^2+\abs{\psi_-(x)}^2\right)\left(\abs{\psi_+(y)}^2+\abs{\psi_-(y)}^2\right)}{\abs{x-y}} \ud x \ud y -
  \alpha \int \left( \abs{\psi_+(x)}^{8/3} + \abs{\psi_-(x)}^{8/3} \right)
  \ud x,
\end{multline}
where $V_R(x)$ is the nuclear potential. The constraints on $(\psi_+,\psi_-)$ are
\begin{equation}
\label{eqn:cons}
\int \abs{\psi_{i}(x)}^2 \ud x=1,i=+,-  .
\end{equation}

We define the Lagrangian as,
\begin{multline}
L(\psi_+,\psi_-,\epsilon_+,\epsilon_-)=\frac{1}{2} \int \abs{\nabla \psi_+}^2
  \ud x + \frac{1}{2} \int \abs{\nabla \psi_-}^2 \ud x
  +   \int V_R(x)\left(\abs{\psi_+(x)}^2+\abs{\psi_-(x)}^2\right)  \ud x \\
  + \frac{1}{2} \iint \cfrac{\left(\abs{\psi_+(x)}^2+\abs{\psi_-(x)}^2\right)\left(\abs{\psi_+(y)}^2+\abs{\psi_-(y)}^2\right)}{\abs{x-y}} \ud x \ud y -
  \alpha \int \left( \abs{\psi_+(x)}^{8/3} + \abs{\psi_-(x)}^{8/3} \right)
  \ud x \\
 - \epsilon_+\left(\int \abs{\psi_+(x)}^2 \ud x-1\right)-\epsilon_-\left(\int \abs{\psi_-(x)}^2 \ud x-1\right),
 \label{lagran}
\end{multline}
where $(\epsilon_+,\epsilon_-)$ are Lagrange multipliers.
 \noindent

Finding the stationary variation of (\ref{lagran}) with respect to the functions $\psi_+$ and $\psi_-$ leads to effective one-electron eigenvalue equations.  
\begin{equation}
\label{eqn:Lag}
\cfrac{\delta L}{\delta\psi_i}=0\Rightarrow\left(-\frac{1}{2}\nabla^2+ V_R(x)+ \int \cfrac{\left(\abs{\psi_+(y)}^2+\abs{\psi_-(y)}^2\right)}{\abs{x-y}} \ud y-\frac{4}{3}\alpha\abs{\psi_{i}(x)}^{2/3}\right)\psi_{i}(x) = \epsilon_{i} \psi_{i}(x), i=\pm,
\end{equation}
where $(\psi_+,\psi_-)$ and $(\epsilon_+,\epsilon_-)$ satisfy normalization constraints. 

In order to determine whether the stationary extremum of $L(\psi_+,\psi_-,\epsilon_+,\epsilon_-)$ with respect to functional variation are a maximum, a minimum or a saddle point, the second order functional derivative (the Hessian matrix) may be analyzed \cite{HoudongThesis}. In the following the stationary solutions $(\psi_+,\psi_-)$ and their eigenvalues $(\epsilon_+,\epsilon_-)$ satisfy  \eqref{eqn:Lag} and \eqref{eqn:cons}. $(\lambda_i, w_i)$ are eigenvalues and eigenvectors of the Hessian Matrix, $\textbf{Hess}$, defined as the solutions to the eigenvalue problem,
\begin{equation}
\textbf{Hess}w_i(y)=\int\cfrac{\delta^2 L(\psi,\epsilon)}{\delta\psi_i(x)\delta\psi_j(y)} w_i(y)\ud y\bigg|_{(\psi,\epsilon)=(\phi,\epsilon)}=\lambda_iw_i(x), \ i=\pm,
\end{equation}
where the  Hessian matrix is defined as,
\begin{equation}
\label{eqn:Hess}
\textbf{Hess}= \left( \begin{array}{cc}
 H_{11} & \int\cfrac{2\psi_-(y)}{\abs{x-y}}\ud y \psi_+(x) \\
\int \cfrac{2\psi_+(y)}{\abs{x-y}}\ud y \psi_-(x) & H_{22}
\end{array}\right)
\end{equation}
where
\begin{equation}
\begin{array}{lcl}
&& H_{11}=-\frac{1}{2}\nabla^2+V_R+\int \cfrac{\abs{\psi_-(y)}^2}{\abs{x-y}}\ud y-\frac{20}{9}\alpha\abs{\psi_+(x)}^{2/3}-\epsilon_+,\nonumber\\
&& H_{22}=-\frac{1}{2}\nabla^2+V_R+\int \cfrac{\abs{\psi_+(y)}^2}{\abs{x-y}}\ud y-\frac{20}{9}\alpha(\abs{\psi_-(x)}^{2/3}-\epsilon_-. \nonumber\\
&&
\end{array}
\end{equation} 
In addition the eigenfunctions $w_i(x)$ satisfy the orthogonality relations,

\begin{equation}
\begin{array}{lcl}
\int \psi_i(x)w_i(x) \ud x=0, i=+,-.
\end{array}
\end{equation}

If all the eigenvalues of $\textbf{Hess}$ are positive, then there is no descent direction in the function space. Negative eigenvalues imply that there is a descent direction.   To carry out this calculation the integrals of the Coulomb potential required in \eqref{eqn:Hess} are obtained by solving the Poisson equation as in \eqref{eq.poi k-1}, and performing the numerical integrals. The calculated eigenvalues are shown in Table~\ref{tab:hess}.

\end{appendix}

\bibliographystyle{plain}
\bibliography{dft}

\begin{thebibliography}{10}

\bibitem{githubFEM}
Finite element basis functions.
\newblock http://hplgit.github.io.

\bibitem{hypre}
Hypre library.
\newblock
  https://computation.llnl.gov/projects/hypre-scalable-linear-solvers-multigrid-methods.

\bibitem{anantharaman2009existence}
Arnaud Anantharaman and Eric Canc{\`e}s.
\newblock Existence of minimizers for {K}ohn--{S}ham models in quantum
  chemistry.
\newblock {\em Annales de l'Institut Henri Poincare (C) Non Linear Analysis},
  26(6):2425--2455, 2009.

\bibitem{apra2020nwchem}
Edoardo Apr\`a, Eric~J. Bylaska, Wibe~A. De~Jong, Niranjan Govind, Karol
  Kowalski, Tjerk~P. Straatsma, Marat Valiev, Hubertus~J.J. van Dam, Yuri
  Alexeev, James Anchell, and et~al.
\newblock {NWChem: Past, present, and future}.
\newblock {\em The Journal of Chemical Physics}, 152(18):184102, 2020.

\bibitem{axelsson2001finite}
Owe Axelsson and V.~Alan Barker.
\newblock {\em Finite Element Solution of Boundary Value Problems: Theory and
  Computation}.
\newblock SIAM, 2001.

\bibitem{bank1981optimal}
Randolph~E. Bank and Todd Dupont.
\newblock An optimal order process for solving finite element equations.
\newblock {\em Mathematics of Computation}, 36(153):35--51, 1981.

\bibitem{barca2014communication}
Giuseppe M.~J. Barca, Andrew T.~B. Gilbert, and Peter M.~W. Gill.
\newblock {H}artree-{F}ock description of excited states of {H}2.
\newblock {\em The Journal of Chemical Physics}, 141(11):111104, 2014.

\bibitem{benguria1981thomas}
Rafael Benguria, Ha{\"\i}m Br{\'e}zis, and Elliott~H. Lieb.
\newblock The {T}homas-{F}ermi-von {W}eizs{\"a}cker theory of atoms and
  molecules.
\newblock {\em Communications in Mathematical Physics}, 79(2):167--180, 1981.

\bibitem{braess2001finite}
Dietrich Braess.
\newblock {\em Finite elements: theory, fast solvers and applications in solid
  mechanics, second edition}.
\newblock Cambridge University Press, 2001.

\bibitem{bramble1993multigrid}
James~H. Bramble.
\newblock {\em Multigrid methods}, volume 294.
\newblock CRC Press, 1993.

\bibitem{bramble1987new}
James~H. Bramble and Joseph~E. Pasciak.
\newblock New convergence estimates for multigrid algorithms.
\newblock {\em Mathematics of Computation}, 49(180):311--329, 1987.

\bibitem{brandt1986algebraic}
Achi Brandt.
\newblock Algebraic multigrid theory: The symmetric case.
\newblock {\em Applied Mathematics and Computation}, 19(1):23--56, 1986.

\bibitem{brandt1985algebraic}
Achi Brandt, Steve McCormick, and John Ruge.
\newblock Algebraic multigrid (amg) for sparse matrix equations.
\newblock {\em Sparsity and its Applications}, page 257, 1985.

\bibitem{brenner2008mathematical}
Susanne~C. Brenner and L.~Ridgway Scott.
\newblock {\em The mathematical theory of finite element methods}, volume~15.
\newblock Springer, 2008.

\bibitem{briggs2000multigrid}
William~L. Briggs, Van~Emden Henson, and Stephen~F. McCormick.
\newblock {\em A multigrid tutorial: second edition}.
\newblock SIAM, 2000.

\bibitem{bylaska2009adaptive}
Eric~J. Bylaska, Michael Holst, and John~H. Weare.
\newblock Adaptive finite element method for solving the exact kohn- sham
  equation of density functional theory.
\newblock {\em Journal of Chemical Theory and Computation}, 5(4):937--948,
  2009.

\bibitem{weare1995parallel}
Eric~J. Bylaska, Scott~R. Kohn, Scott~B. Baden, Alan Edelman, Ryoichi Kawai,
  M.~Elizabeth~G. Ong, and John~H. Weare.
\newblock Scalable parallel numerical methods and software tools for material
  design.
\newblock {\em Proceedings of the Seventh SIAM Conference on Parallel
  Processing for Scientific Computing}, pages 219--224, 1995.

\bibitem{kubicki}
Ying Chen, Eric Bylaska, and John Weare.
\newblock First principles estimation of geochemically important transition
  metal oxide properties.
\newblock In James~D. Kubicki, editor, {\em Molecular Modeling of Geochemical
  Reactions}, chapter~4. Wiley, 2016.

\bibitem{Yang:08}
Aron~J. Cohen, Paula Mori-S{\'a}nchez, and Weitao Yang.
\newblock Insights into current limitations of density functional theory.
\newblock {\em Science}, 321:792 -- 794, 2008.

\bibitem{Yang:12}
Aron~J. Cohen, Paula Mori-S{\'a}nchez, and Weitao Yang.
\newblock Challenges for density functional theory.
\newblock {\em Chemical Reviews}, 112:289--320, 2012.

\bibitem{claude91}
Claude Cohen-Tannoudji, Bernard Diu, and Frank Laloe.
\newblock {\em Quantum Mechanics, Volume 1}.
\newblock Wiley, 1991.

\bibitem{Cox1992}
Paul~A. Cox.
\newblock {\em Transition Metal Oxides}.
\newblock Clearendon Press, Oxford, 1992.

\bibitem{dupont2003fenics}
Todd Dupont, Johan Hoffman, Claus Johnson, Robert~C. Kirby, Mats~G. Larson,
  Anders Logg, and L.~Ridgway Scott.
\newblock {\em The FEniCS project}.
\newblock Chalmers Finite Element Centre, Chalmers University of Technology,
  2003.

\bibitem{gauss}
James~B. Foresman and Aeleen Frisch.
\newblock {\em Exploring Chemistry with Electronic Structure Methods}.
\newblock Gaussian, 1996.

\bibitem{FL2}
Rupert~L. Frank, Elliott~H. Lieb, Robert Seiringer, and Lawrence~E. Thomas.
\newblock Bipolaron and n-polaron binding energies.
\newblock {\em Physical Review Letters}, 104(21):210402, 2010.

\bibitem{FL1}
Rupert~L. Frank, Elliott~H. Lieb, Robert Seiringer, and Lawrence~E. Thomas.
\newblock Stability and absence of binding for multi-polaron systems.
\newblock {\em Publications math{\'e}matiques de l'IH{\'E}S}, 113(1):39--67,
  2011.

\bibitem{golub}
Gene~H. Golub.
\newblock {\em Matrix Computations}.
\newblock Johns Hopkins Press, 1996.

\bibitem{gontier2015existence}
David Gontier.
\newblock Existence of minimizers for {K}ohn--{S}ham within the local spin
  density approximation.
\newblock {\em Nonlinearity}, 28(1):57, 2015.

\bibitem{gontier2018lower}
David Gontier, Christian Hainzl, and Mathieu Lewin.
\newblock Lower bound on the {H}artree-{F}ock energy of the electron gas.
\newblock {\em arXiv preprint arXiv:1811.12461}, 2018.

\bibitem{gontier2018spin}
David Gontier and Mathieu Lewin.
\newblock Spin symmetry breaking in the translation-invariant {H}artree-{F}ock
  uniform electron gas.
\newblock {\em arXiv preprint arXiv:1812.07679}, 2018.

\bibitem{GriesemerHantsch:12}
Marcel Griesemer and Fabian Hantsch.
\newblock Unique solutions to {H}artree--{F}ock equations for closed shell
  atoms.
\newblock {\em Arch. Rational Mech. Anal.}, 203(3):883--900, 2012.

\bibitem{hackbusch1985multi}
Wolfgang Hackbusch.
\newblock {\em Multi-grid methods and applications}, volume~4.
\newblock Springer-Verlag Berlin, 1985.

\bibitem{harrison2004multiresolution}
Robert~J. Harrison, George~I. Fann, Takeshi Yanai, Zhengting Gan, and Gregory
  Beylkin.
\newblock Multiresolution quantum chemistry: Basic theory and initial
  applications.
\newblock {\em The Journal of Chemical Physics}, 121(23):11587--11598, 2004.

\bibitem{hehre1969self}
Warren~J. Hehre, Robert~F. Stewart, and John~A. Pople.
\newblock self-consistent molecular-orbital methods. i. use of gaussian
  expansions of slater-type atomic orbitals.
\newblock {\em The Journal of Chemical Physics}, 51(6):2657--2664, 1969.

\bibitem{HoudongThesis}
Houdong Hu.
\newblock Electronic structure models: {S}olution theory, linear scaling
  methods, and stability analysis.
\newblock {\em UCSD Ph.D. Thesis}, 2014.

\bibitem{kendall2000high}
Ricky~A. Kendall, Edoardo Apr{\`a}, David~E. Bernholdt, Eric~J. Bylaska, Michel
  Dupuis, George~I. Fann, Robert~J. Harrison, Jialin Ju, Jeffrey~A. Nichols,
  Jarek Nieplocha, et~al.
\newblock High performance computational chemistry: An overview of nwchem a
  distributed parallel application.
\newblock {\em Computer Physics Communications}, 128(1):260--283, 2000.

\bibitem{kirr2011symmetry}
Eduard Kirr, Panayotis~G. Kevrekidis, and Dmitry~E. Pelinovsky.
\newblock Symmetry-breaking bifurcation in the nonlinear {S}chr{\"o}dinger
  equation with symmetric potentials.
\newblock {\em Communications in Mathematical Physics}, 308(3):795--844, 2011.

\bibitem{weare1997adaptive}
Scott~R. Kohn, John~H. Weare, M.~Elizabeth~G. Ong, and Scott~B. Baden.
\newblock Parallel adaptive mesh refinement for electronic structure
  calculations.
\newblock {\em Proceedings of the Eighth SIAM Conference on Parallel Processing
  for Scientific Computing}, 1997.

\bibitem{le1993some}
Claude Le~Bris.
\newblock Some results on the {T}homas-{F}ermi-{D}irac-von {W}eizs{\"a}cker
  model.
\newblock {\em Differential and Integral Equations}, 6(2):337--353, 1993.

\bibitem{Lenzmann}
Enno Lenzmann.
\newblock Uniqueness of ground states for pseudorelativistic {H}artree
  equations.
\newblock {\em Analysis \& PDE}, 2(1):1--27, 2009.

\bibitem{lieb1981thomas}
Elliott~H Lieb.
\newblock Thomas-{F}ermi and related theories of atoms and molecules.
\newblock {\em Reviews of Modern Physics}, 53(4):603, 1981.

\bibitem{LiebLoss:01}
Elliott~H. Lieb and Michael Loss.
\newblock {\em Analysis}.
\newblock American Mathematical Society, 2nd edition, 2001.

\bibitem{lieb1977hartree}
Elliott~H Lieb and Barry Simon.
\newblock The {H}artree-{F}ock theory for {C}oulomb systems.
\newblock {\em Communications in Mathematical Physics}, 53(3):185--194, 1977.

\bibitem{lions1987solutions}
Pierre-Louis Lions.
\newblock Solutions of {H}artree-{F}ock equations for {C}oulomb systems.
\newblock {\em Communications in Mathematical Physics}, 109(1):33--97, 1987.

\bibitem{fenics2012}
Anders Logg, Kent-Andre Mardel, and Garth~N. Wells, editors.
\newblock {\em Automated Solution of Differential Equations by the Finite
  Element method: The FEniCS Book}.
\newblock Springer, 2012.

\bibitem{dolfin}
Anders Logg and Garth~N. Wells.
\newblock Dolfin: Automated finite element computing.
\newblock {\em ACM Transations on Mathematical Software (TOMS)}, 37(issue 2,
  article 20), 2010.

\bibitem{mccormick1987multigrid}
Stephen~F. McCormick.
\newblock {\em Multigrid methods}, volume~3.
\newblock SIAM, 1987.

\bibitem{Oliver-Perdew}
G.~L. Oliver and J.~P. Perdew.
\newblock Spin-density gradient expansion for kinetic energy.
\newblock {\em Physical Review A}, 20(2):397--403, 1979.

\bibitem{Parr1972}
Robert~G. Parr.
\newblock {\em Quantum Theory of Molecular Electronic Structure}.
\newblock W. A. Benjamin, 1972.

\bibitem{Parr-Yang}
Robert~G. Parr and Weitao Yang.
\newblock {\em Density-Functional Theory of Atoms and Molecules}.
\newblock Oxford Science Publications, 1989.

\bibitem{Peng-Perdew2017}
Haowei Peng and John~P. Perdew.
\newblock Synergy of van der waals and self-interaction corrections in
  transition metal monoxides.
\newblock {\em Physical Review B}, 96:100101 1--5, 2017.

\bibitem{Henkelman2011}
Zachary~D. Pozun and Graeme Henkelman.
\newblock Hybrid density functional theory band structure engineering in
  hematite.
\newblock {\em The Journal of Chemical Physics}, 134:224706--1--9, 2011.

\bibitem{RS4}
Michael Reed and Barry Simon.
\newblock {\em Analysis of Operators, Vol. IV of Methods of Modern Mathematical
  Physics}.
\newblock New York, Academic Press, 1978.

\bibitem{ricaud2017}
Julien Ricaud.
\newblock {\em Sym{\'e}trie et brisure de sym{\'e}trie pour certains
  probl{\`e}mes non lin{\'e}aires}.
\newblock PhD thesis, Universit{\'e} de Cergy Pontoise, 2017.

\bibitem{ricaud2017symmetry}
Julien Ricaud.
\newblock Symmetry breaking in the periodic {T}homas--{F}ermi--{D}irac--von
  {W}eizs{\"a}cker model.
\newblock {\em Annales Henri Poincar{\'e}}, 19(10):3129--3177, 2018.

\bibitem{Roll2004}
Georg Rollmann, Alexander Rohrbach, Peter Entel, and J\"urgen Hafner.
\newblock First-principle calculations of the structure and magnetic properties
  of hematite.
\newblock {\em Physical Review B}, 69:165107 1--12, 2004.

\bibitem{ruiz2010schrodinger}
David Ruiz.
\newblock On the {S}chr{\"o}dinger--{P}oisson--{S}later system: behavior of
  minimizers, radial and nonradial cases.
\newblock {\em Archive for Rational Mechanics and Analysis}, 198(1):349--368,
  2010.

\bibitem{RuskaiStillinger:84}
Mary~Beth Ruskai and Frank~H Stillinger.
\newblock Binding limit in the {H}artree approximation.
\newblock {\em Journal of Mathematical Physics}, 25(6):2099--2103, 1984.

\bibitem{gmres}
Yousef Saad and Martin~H. Schultz.
\newblock Gmres: A generalized minimal residual algorithm for solving
  nonsymmetric linear systems.
\newblock {\em SIAM J. Sci. Stat. Comput.}, 7(3):856–869, 1986.

\bibitem{SS}
Catherine Sulem and Pierre-Louis Sulem.
\newblock {\em The nonlinear {S}chr{\"o}dinger equation: self-focusing and wave
  collapse}, volume 139.
\newblock Springer Science \& Business Media, 1999.

\bibitem{Szabo1989}
Attila Szabo and Neil~S. Ostlund.
\newblock {\em Modern Quantum Chemistry}.
\newblock Dover Publications, 1989.

\bibitem{Tatebe93}
Osamu Tatebe.
\newblock The multigrid preconditioned conjugate gradient method.
\newblock {\em NASA. Langley Research Center, The Sixth Copper Mountain
  Conference on Multigrid Methods, Part 2}, 1993.

\bibitem{valiev2010nwchem}
Marat Valiev, Eric~J. Bylaska, Niranjan Govind, Karol Kowalski, Tjerk~P.
  Straatsma, Hubertus~J.J. Van~Dam, Dunyou Wang, Jarek Nieplocha, Edoardo
  Apr\`a, Theresa~L. Windus, and et~al.
\newblock Nwchem: a comprehensive and scalable open-source solution for large
  scale molecular simulations.
\newblock {\em Computer Physics Communications}, 181(9):1477--1489, 2010.

\bibitem{Boomer2002}
Henson Van~Emden and Ulrike~Meir Yang.
\newblock Boomer amg: A parallel algebraic multigrid solver and preconditioner.
\newblock {\em Applied Numerical Mathematics}, 41:155--177, 2002.

\bibitem{Weinstein}
Michael~I Weinstein.
\newblock Lyapunov stability of ground states of nonlinear dispersive evolution
  equations.
\newblock {\em Communications on Pure and Applied Mathematics}, 39(1):51--67,
  1986.

\end{thebibliography}

\end{document}